\documentclass[lettersize,journal]{IEEEtran}

\usepackage{cite}
\usepackage[numbers,sort&compress]{natbib}
\usepackage{amsmath,amssymb,amsfonts}
\usepackage{textcomp}
\usepackage[dvipsnames]{xcolor}
\usepackage{graphicx,subfigure}
\usepackage{multirow}
\usepackage{array}
\usepackage{bm}
\usepackage{enumerate}
\usepackage{indentfirst}
\usepackage{amsthm}
\usepackage{booktabs}
\usepackage[most]{tcolorbox}
\usepackage[ruled, linesnumbered]{algorithm2e}
\usepackage[justification=centering]{caption}
\newtheorem{theorem}{\textbf{Theorem}}

\theoremstyle{definition}
\newtheorem{definition}{Definition}

\begin{document}

\title{Efficient Blockchain-based Steganography via Backcalculating Generative Adversarial Network}

\author{Zhuo~Chen,
        Jialing~He,
        Jiacheng~Wang,
        Zehui Xiong,~\IEEEmembership{Senior~Member,~IEEE,}
        Tao~Xiang,~\IEEEmembership{Senior~Member,~IEEE,}
        Liehuang~Zhu,~\IEEEmembership{Senior~Member,~IEEE,}
        and~Dusit~Niyato,~\IEEEmembership{Fellow,~IEEE}
\thanks{Zhuo Chen and Liehuang Zhu are with the School of Cyberspace Science and Technology, Beijing Institute of Technology, Beijing 100081, China.}
\thanks{Jialing He and Tao Xiang are with the College of Computer Science, Chongqing University, Chongqing 400044, China.}
\thanks{Jiacheng Wang and Dusit Niyato are with Nanyang Technological University, Singapore.}
\thanks{Zehui Xiong is with the School of Electronics, Electrical Engineering and Computer Science (EEECS), Queen's University Belfast, Belfast, BT7 1NN, U.K..}
\thanks{Corresponding author: Jiacheng Wang.}
}




\maketitle

\begin{abstract}
  Blockchain-based steganography enables data hiding via encoding the covert data into a specific blockchain transaction field. However, previous works focus on the specific field-embedding methods while lacking a consideration on required field-generation embedding. In this paper, we propose a generic blockchain-based steganography framework (GBSF). The sender generates the required fields such as amount and fees, where the additional covert data is embedded to enhance the channel capacity. Based on GBSF, we design a reversible generative adversarial network (R-GAN) that utilizes the generative adversarial network with a reversible generator to generate the required fields and encode additional covert data into the input noise of the reversible generator. We then explore the performance flaw of R-GAN. To further improve the performance, we propose R-GAN with \underline{C}ounter-intuitive data preprocessing and \underline{C}ustom activation functions, namely \underline{CC}R-GAN. The counter-intuitive data preprocessing (CIDP) mechanism is used to reduce decoding errors in covert data, while it incurs gradient explosion for model convergence. The custom activation function named ClipSigmoid is devised to overcome the problem. 
  Theoretical justification for CIDP and ClipSigmoid is also provided. 
  We also develop a mechanism named T2C, which balances capacity and concealment.
  We conduct experiments using the transaction amount of the Bitcoin mainnet as the required field to verify the feasibility. 
  We then apply the proposed schemes to other transaction fields and blockchains to demonstrate the scalability. 
  Finally, we evaluate capacity and concealment for various blockchains and transaction fields and explore the trade-off between capacity and concealment. The results demonstrate that R-GAN and CCR-GAN are able to enhance the channel capacity effectively and outperform state-of-the-art works.
\end{abstract}


\section{Introduction}

Steganography enables both senders and receivers to transmit data secretly over a public network channel~\cite{mandal2022digital, liu2022deep, zheng2022wmdefense}. It is widely used in digital watermarking~\cite{qiang2023natural}, censorship-resistant systems~\cite{rosen2021balboa} and digital forensics~\cite{zhou2023generative}. Due to concealing the communication behavior between the sender and the receiver, steganography ensures a secure transmission of confidential military and commercial information~\cite{xu2022robust, wang2024generative2}. In the blockchain-based steganography~\cite{zhang2023covert, miao2022privacy}, the sender and the receiver establish a covert channel through the blockchain network~\cite{chen2022blockchain, du2022applications}. The sender hides the covert data into a specific transaction field and broadcasts the covert transaction to the blockchain. The receiver identifies the covert transaction from the blockchain and decodes it to access the covert data.

Previous studies primarily focus on encoding the covert data into specific transaction fields (i.e., embedding fields), while rarely considering the proper generation of required fields that used to complete a transaction~\cite{wang2023covert, zhang2023ebdl}. Wang~\textit{et al.}~\cite{wang2022practical} demonstrated that the improper generation of the required fields can easily expose covert transactions, and proposed a required-field-generation method by applying generative adversarial networks (GANs)~\cite{creswell2018generative, wang2024generative}. This approach is able to generate required fields that are indistinguishable from normal transaction fields. However, it lacks a consideration of embedding the covert data into the required fields' generation process, facing the following main challenges. 
\begin{itemize}
  \item \emph{Less redundancy.} Blockchain fields contain less redundant information than those of audio and image data~\cite{koptyra2022extension, zhang2020magview}, making it more challenging to conceal data.
  \item \emph{No semantics.} Unlike text, a blockchain field is simply a number without semantic information. As a result, typical text steganography methods~\cite{yang2023semantic, yang2022linguistic, li2022gasla, cui2019unseencode}, which rely on semantic and state transfer probabilities, are unsuitable for blockchain field steganography.
  \item \emph{Difficult to encode.} Existing studies often employ deep generative models to generate indistinguishable required fields~\cite{wang2022practical}. However, these models are often uninterpretable, making it difficult to encode data into the generated fields.
\end{itemize}

In this paper, we propose a generic blockchain-based steganography framework (GBSF) that improves the capacity of blockchain-based covert channels. In GBSF, the sender employs a GAN to generate required fields and encodes the covert data as input to the GAN's generator. This input consists of random noise, which is similar to the random string nature of covert data (often ciphertext). However, as deep learning models are typically irreversible, it is difficult for the receiver to restore the covert data. To address this challenge, we therefore introduce the concept of a reversible GAN (R-GAN) whose generator is reversible. With the reversible GAN, we then introduce the R-GAN scheme. Our key insight is to model the generator of GAN as an invertible multivariate function by carefully configuring the network structure. We accomplish this by constructing the generator using linear neural networks (e.g., fully connected layers and convolutional neural networks~\cite{he2024preventing}) and reversible activation functions (e.g., Sigmoid and LeakyReLU~\cite{he2025advancing}). The receiver backcalculates the generator to retrieve the input noise and the covert data.

R-GAN is only capable of embedding and recovering a small amount of covert data. This is because the generator of R-GAN tends to produce fractional numbers, whereas the transaction field requires integer values. Consequently, it incurs inaccuracy and rounding errors. To enhance performance, we introduce R-GAN with a \underline{C}ounter-intuitive data preprocessing (CIDP) method and a \underline{C}ustom activation function, namely \underline{CC}R-GAN. We propose CIDP to mitigate the rounding error. We design a custom activation function called ClipSigmoid as the output of the model to eliminate excessive gradient and to improve model convergence. Additionally, we provide theoretical justification for both CIDP and ClipSigmoid.
Inspired by CIDP, we also propose T2C, a mechanism for balancing data embedding capacity and concealment. The core idea of T2C is to improve the quality of the training dataset at the expense of increasing the rounding error, thus trading data embedding capacity for concealment.
\begin{figure*}[t]
\centerline{\includegraphics[width=0.9\textwidth]{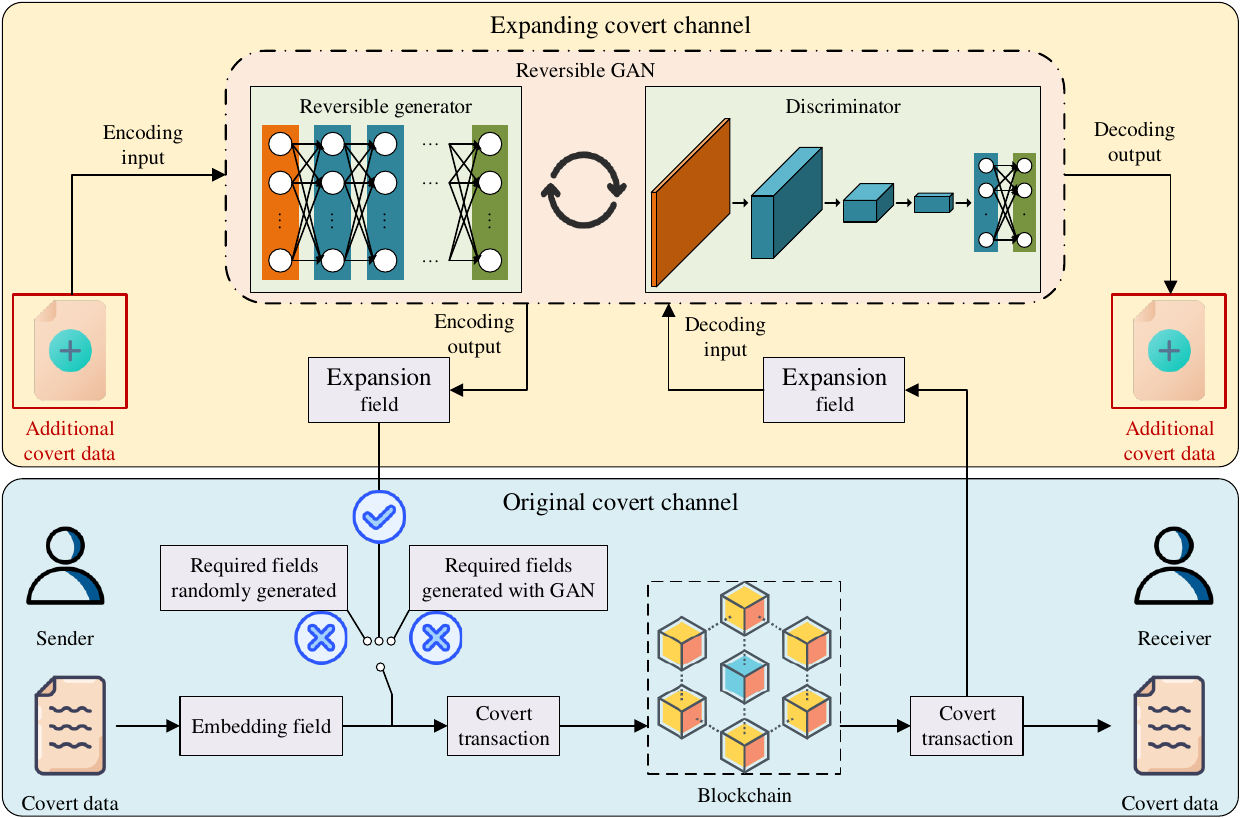}}
\captionsetup{justification=raggedright}
\caption{Overview of GBSF framework. In GBSF, the sender uses the R-GAN to generate the required fields and complete a transaction. The generator of R-GAN is reversible. It inputs additional covert data and outputs transaction fields which we call expansion fields. Given the expansion field, the receiver can calculate the generator in reverse to extract covert data.}
\label{fig: system_model}
\end{figure*}

Finally, we present an implementation of both R-GAN and CCR-GAN and verify their scalability. We first use the Bitcoin transaction amount as an example to evaluate the feasibility of proposed schemes. We then employ Bitcoin fee, Ethereum amount, and Ethereum fee as the dataset, and test the data embedding ability and concealment of R-GAN/CCR-GAN when generating these fields to evaluate the scalability. Experimental results reveal that both R-GAN and CCR-GAN are able to be applied to the above blockchains and transaction fields. Furthermore, we explore the capacity and concealment aspects of R-GAN and CCR-GAN. Compared to R-GAN, CCR-GAN is able to embed a larger amount of covert data at the expense of partial concealment since it has a lower rounding error. When the rounding error of CCR-GAN is small enough, the computational precision error becomes a decisive factor in limiting the amount of embedded data. We further explore the edges of the rounding error and the computational precision error. The experimental results show that under the IEEE 754 standard~\cite{guo2024single}, a dual-layer generator network supports embedding up to approximately 40 bits of data in a transaction field due to the computational precision error. The actual embedded data amount depends on the rounding error, which is determined by the magnitude of the difference between the maximum and minimum values in the training dataset. A larger magnitude results in smaller rounding errors and allows for more data to be embedded. When the magnitude reaches $10^{17}$, the rounding error can be considered sufficiently small, and the amount of embedded data reaches the upper limit imposed by the computational precision error. Overall, both the data embedding capacity and concealment of the proposed schemes surpass those of baselines.

In summary, main contributions of this paper include:
\begin{itemize}
  \item We propose a GBSF framework that enhances the capacity of blockchain-based covert channels. In GBSF, the sender leverages the reversible GAN to generate the required fields for creating transactions, while encoding covert data as the input to the generator. In this way, GBSF can effectively improve channel capacity.
  \item We propose two schemes, namely R-GAN and CCR-GAN. The main concept behind R-GAN is to model its generator as a reversible function that allows the receiver to decode covert data by reversing the generator's computation. However, R-GAN has a capacity limitation due to rounding errors. We further propose CCR-GAN to improve capacity. In comparison to R-GAN, CCR-GAN incorporates a counter-intuitive data preprocessing method named CIDP and a custom activation function called ClipSigmoid. These additions reduce rounding errors and facilitate model convergence. We also present analytical justifications for CIDP and ClipSigmoid.
  \item We utilize the Bitcoin transaction amount as the required field for implementing and evaluating GBSF on the Bitcoin mainnet. We conduct tests to measure the amount of covert data that R-GAN and CCR-GAN can embed in each transaction amount, as well as their capacity to enhance existing blockchain-based steganography schemes. Experimental results demonstrate that both schemes exhibit a high level of capacity and concealment enhancement capabilities.
  \item We evaluate the scalability of R-GAN and CCR-GAN. We apply R-GAN and CCR-GAN to the transaction fee field of Bitcoin. Experimental results show that the Bitcoin fee is also applicable to the proposed schemes. We also apply R-GAN and CCR-GAN to Ethereum. The experimental results show that the proposed schemes can also be extended to Ethereum amount and Ethereum fee.
  \item We devise T2C, a mechanism to support fine-grained trading capacity for concealment. We also design experiments to verify the effectiveness of T2C. The results show that for every reduction/increase in capacity by 2-3 bits, concealment is simultaneously increased/decreased by 3\%-4\%. We also find that there is a boundary in the trade off between capacity and concealment. When capacity reaches the upper limit, it is no longer possible to increase capacity by sacrificing concealment. Under the IEEE 754 standard, the capacity limit is about 40 bits per transaction field.
\end{itemize}

\section{Covert Channel: R-GAN}

\subsection{GBSF Framework}

We introduce a generic framework for enhancing the capacity of blockchain-based steganography, as illustrated in Fig.~\ref{fig: system_model}. The framework consists of an original covert channel and an expanding covert channel. 

The blue box illustrates the original covert channel, consisting of a sender, a receiver, and a blockchain. The sender encodes the covert data into a transaction field (i.e., the embedding field) such as the address and the signature~\cite{yuan2025blockchain,chen2025tackling}. Afterwards, the sender generates the remaining transaction fields required for transaction creation (i.e., the required field), either randomly or using deep generative models. These required fields are utilized to create the covert transaction, which is then broadcasted to the blockchain network. The receiver retrieves the covert transaction from the blockchain and decodes the covert data based on the embedding field.

The expanding covert channel (depicted in the yellow box in Fig. \ref{fig: system_model}) enhances the channel capacity, which is achieved by embedding additional covert data in required fields rather than generating required fields randomly or using deep generative models. The key concept behind the expanding covert channel is the utilization of a specialized GAN, i.e., R-GAN. Similar to a typical GAN, the R-GAN consists of a generator and a discriminator. The difference is that its generator is reversible, allowing the function expression of the generator to be backcalculated. This property enables the sender to encode additional covert data using the generator in the forward direction, while the receiver utilizes the generator in the back direction to decode the additional covert data. With R-GAN, the sender can encode additional covert data into transaction fields that are indistinguishable from normal transaction fields. These fields carrying the additional covert data are known as expansion fields. By incorporating the expansion fields alongside the original embedding fields, the sender constructs the complete covert transactions, effectively increasing the capacity of the blockchain-based covert channel. The receiver retrieves the expansion field from the covert transaction and inputs it into the R-GAN. The R-GAN then performs the reverse process of encoding and outputs the additional covert data. In this way, we establish an expanding covert channel to transmit the additional covert data and boost the channel capacity.
In the following, we present the technical details of the expanding covert channel with R-GAN. 

\subsection{Technical Overview}

As depicted in Fig. \ref{fig: scheme_process}, the expanding covert channel consists of three main steps. Firstly, the sender and the receiver train a model to generate expansion fields. Secondly, the sender encodes covert data into the generated expansion fields. Finally, the receiver decodes the expansion fields to recover the covert data. Note that this paper focuses on details of the encoding and decoding principle and does not consider the model synchronization between the sender and the receiver. Model synchronization method used in~\cite{zhou2022secret,qi2024provably} can also be adopted in our schemes. More precisely, the sender and the receiver can simply obtain an identical model by sharing the rules for training the model. For example, every 3 days, they select the expansion field from the last 100 blocks as a dataset to train the model using a series of identical seeds. The training process does not stop until the model loss falls below a certain threshold.
\begin{figure}[t]
\centerline{\includegraphics[width=0.4\textwidth]{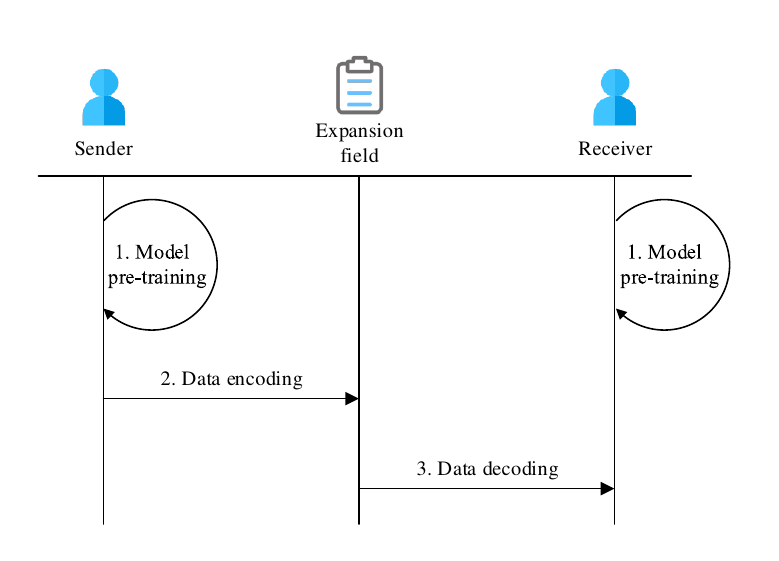}}
\captionsetup{justification=raggedright}
\caption{General workflow of the expanding covert channel. The sender and the receiver first train a R-GAN model, respectively. The sender uses the trained model to encode covert into the expansion field, and the receiver uses the trained model to decode covert data from the expansion field.}
\label{fig: scheme_process}
\end{figure}

\begin{figure}[t]
  \centering
  \subfigure[Transaction amount.]{
      \begin{minipage}[t]{0.45\linewidth}
          \centering
      \includegraphics[width=1.7in]{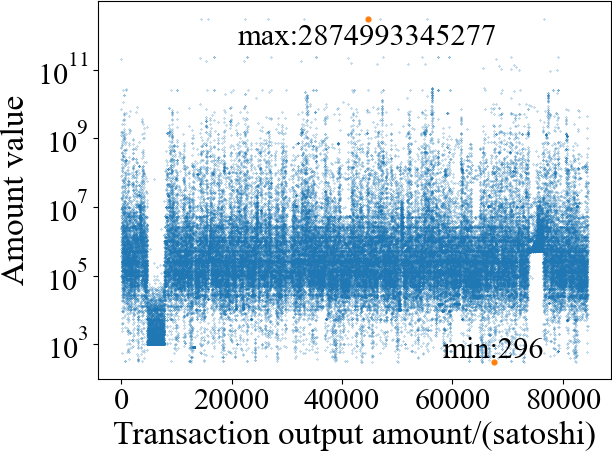}
      \label{fig: tx_amount}
  \end{minipage}
  }
  \subfigure[Length of transaction amount.]{
      \begin{minipage}[t]{0.45\linewidth}
          \centering
      \includegraphics[width=1.7in]{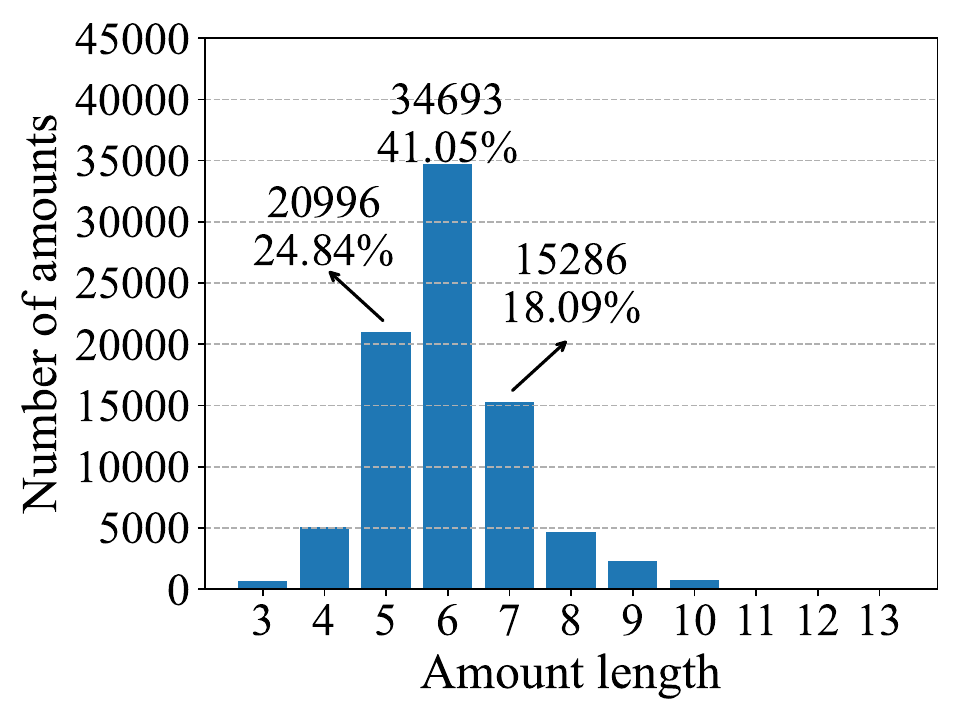}
      \label{fig: len_tx_amount}
      \end{minipage}
  }
  \caption{Example of Bitcoin transaction amount. The transaction amount is a one-dimensional numerical value, ranging in length from 3-13 and concentrated between 5-7.}
  \label{fig: example_tx_amount}
\end{figure}

\subsection{Model Pre-training} \label{sec: basic_dp}

This section outlines the process of model pre-training, which encompasses data preprocessing, model structure, and the loss function. 
We utilize the transaction output amount of Bitcoin as the expansion field to illustrate the pre-training process. We specifically choose this amount field since each transaction must include a specified output amount, and the transaction creator has full control over the output amount.

\subsubsection{Data preprocessing}

\begin{figure*}[t]
\centerline{\includegraphics[width=0.9\textwidth]{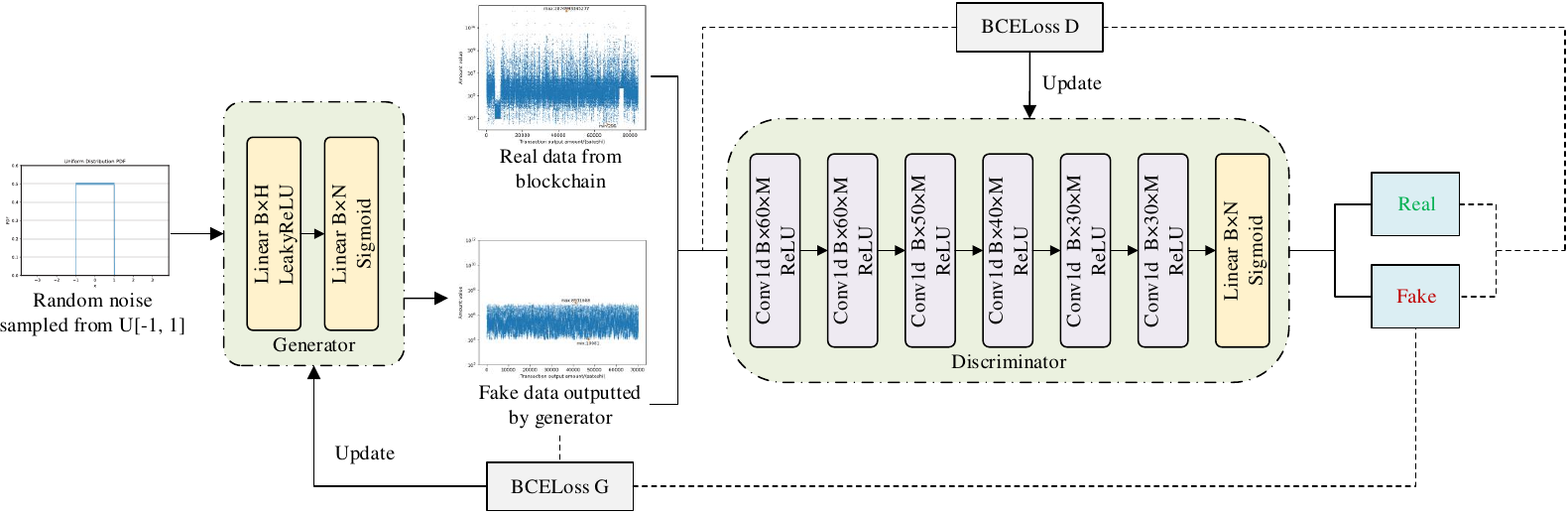}}
\captionsetup{justification=raggedright}
\caption{Overview of R-GAN. $U$ refers to a uniform distribution; $B$ is the batchsize; $H$ represents the hidden dimension; $M$ denotes the input dimension; $N$ is the output dimension. BCELoss refers to the binary cross entropy loss.}
\label{fig: GAN_model}
\end{figure*}
The model takes the transaction output amount as the input, which is represented as an integer like ``18105990". Fig. \ref{fig: tx_amount} illustrates the distribution of Bitcoin transaction output amounts for 25 blocks, from block 727215 to block 727239. These amounts range from 296 to 2,874,993,345,277, with a dense distribution between $10^4$ and $10^7$. Fig. \ref{fig: len_tx_amount} displays the number and proportion of amounts grouped by length. Amounts with lengths ranging from 5 to 7 account for a significant portion, totaling 83.98\%. We select the most common data, specifically the data with lengths ranging from 5 to 7, as the training dataset $X$, and apply min-max normalization to normalize the training dataset to the range from 0 to 1. Let us denote $X = \{x_1, x_2, \cdots, x_n\}$, then for $i \in \{1, 2, \cdots, n\}$, we normalize
\begin{equation}
  x_i' = \frac{x_i - min(X)}{max(X) - min(X)},
\end{equation}
where $min(X)$ and $max(X)$ represent the minimum and maximum elements of $X$, respectively. We also reshape the dataset and group adjacent 64 data points to form a new training dataset. The formation process is represented as:
\begin{equation}
  new\_x_i = \{x_{(i-1)*64 + 1}', x_{(i-1)*64 + 2}', \cdots, x_{i*64}'\},
\end{equation}
where $i \in \{1, 2, \cdots, \lceil n/64 \rceil \}$. 
This formation process helps the model capture the relationship between transaction amounts and reduces the likelihood of encoding the same covert data as the same transaction amount.

\subsubsection{Model structure}

Fig. \ref{fig: GAN_model} presents the overview of R-GAN, comprising a generator and a discriminator. The generator takes a random noise sampled from $U[-1, 1]$ (uniform distribution) as the input and fake amounts as the output. We use a uniform distribution to sample the input noise to accommodate covert data within the noise. The reason we choose the uniform distribution as the input is that it aligns with the randomness nature of encrypted data. The covert data (typically encrypted) encoded into the generator's input can be seen as a random string, which exhibits randomness as well as uniform distribution~\cite{katz2007introduction}. The generator consists of two fully connected layers to ensure reversibility. While the fully connected layer is a simple neural network, it is sufficient to handle one-dimensional numerical data like the transaction amount. The first layer utilizes LeakyReLU as the activation function due to its piecewise linearity property, which can reduces decoding errors during covert data recovery. The last layer employs Sigmoid as the activation function, as it is a monotonic function that outputs values between 0 and 1. The monotonicity of Sigmoid enables reversibility, and its output range aligns with normalized training data. 

The discriminator focuses on evaluating the concealment of generated fake data and is unrelated to the recovery of the covert data. Complex irreversible networks can be employed to construct the discriminator. Due to the excellent feature capture capabilities of convolutional neural networks (CNNs), we use 6 CNNs with ReLU to form the discriminator. The last layer of the discriminator is a fully connected layer with Sigmoid, commonly utilized for binary classification tasks.

\subsubsection{Loss function}

We use the binary cross-entropy loss (BCELoss) to calculate the loss function since BCELoss is more suitable for handling values between 0 and 1 and processing binary classification tasks. The loss function of the model comprises two parts, including the discriminator's loss and the generator's loss. The discriminator's loss is defined as the prediction error of the discriminator on real and fake samples. Assume that the label for real samples is 1 and the label for fake samples is 0, then the loss function of the discriminator can be expressed as follows:
\begin{equation}
  J_{discriminator} = -\frac{1}{2}\sum (\log (1 - D(G(noise))) + \log D(x')),
\end{equation}
where $D$ represents the discriminator, $G$ is the generator, $noise$ denotes the input noise of $G$, and $x'$ represents the normalized element of the training dataset $X$. 

The second part penalizes the generator for generating fake samples that are recognized by the discriminator, which we measure by:
\begin{equation}
  J_{generator} = -\sum \log D(G(noise)).
\end{equation}

\subsection{Data Encoding}

In our schemes, the data encoding process consists of two phases: embedding and verification. In the embedding phase, covert data is embedded into the noise, while the verification phase ensures that the receiver can correctly recover the covert data in the presence of the rounding and computational errors.

\noindent \textbf{Embedding phase.} The sender incorporates the covert data into the input of the generator, which is a noise vector. The sender replaces certain bits of each noise element with covert data. Fig. \ref{fig: encode} illustrates the structure of the noise, which consists of $1$-bit most significant bit (MSB), $m$-bit covert data, and $(n-m-1)$-bit Padding. Both MSB and Padding are uniformly sampled, ensuring the noise spans a wide range of values.
\begin{figure}[t]
\centerline{\includegraphics[width=0.45\textwidth]{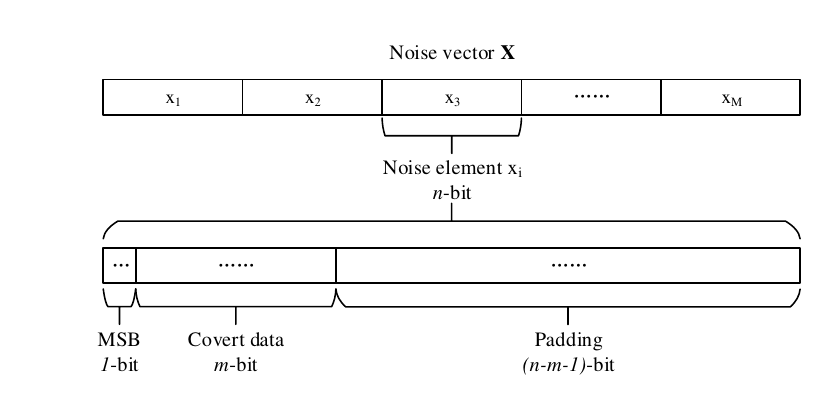}}
\captionsetup{justification=raggedright}
\caption{Noise structure. MSB is the most significant bit.}
\label{fig: encode}
\end{figure}

\begin{algorithm}[htb]
  \caption{Data encoding.}
  \label{alg: encoding}
  \KwIn{Covert data $\mathbf{CD} = (cd_1, cd_2, \cdots, cd_M) \in \{0, 1\}^{m \times M}$.}
  \KwOut{Noise $\mathbf{X} = (x_1, x_2, \cdots, x_M) \in \mathbb{R}^M$.}

  Initialize $\mathbf{X} = (x_1, x_2, \cdots, x_M)$\;
  \While{True}{
    // \textbf{Begin embedding phase}\;
    \For(Randomize $MSB$ and $Padding$){$i$ in $range(\mathsf{M})$}{
      Uniformly sample an $MSB_i \in \{0, 1\}$\;
      Uniformly sample a $Padding_i \in \{0, 1\}^{n-m-1}$\;
      Set $x_i = MSB_i || cd_i || Padding_i$\;
    }
    Compute $\mathbf{Y} = G(\mathbf{X})$\;
    Denormalize $\mathbf{A} = \mathbf{Y} \times (max(X) - min(X)) + min(X) $\;
    // \textbf{Begin verification phase}\;
    Round $\hat{\mathbf{A}} = [\mathbf{A}]$\;
    Normalize $\hat{\mathbf{Y}} = \frac{\hat{\mathbf{A}} - min(X)}{max(X) - min(X)}$\;
    Compute $\hat{\mathbf{X}} = G^{-1}(\hat{\mathbf{Y}})$\;
    Initialize $Result$ = []\;
    \For{$\hat{x}_i \in \hat{\mathbf{X}}$}{
      \If{bits $2$ to $(m+1)$ of $\hat{x}_i$ == bits $2$ to $(m+1)$ of $x_i$}{
        $Result.append(True)$\;
      }
      \Else{
        $Result.append(False)$\;
      }
    }
    \If{All elemets of $Result$ are True}{
      Break\;
    }
  }
  \Return{$\mathbf{X}$}
\end{algorithm}
\noindent \textbf{Verification phase.} The verification phase ensures that the bits representing the covert data in the recovered noise match those sent by the sender. This phase is necessary because the receiver may not obtain the exact same noise as the sender during the decoding process. The inconsistencies may arise due to computational errors and rounding errors. The computational errors refer to inaccuracies in decimal calculations on computers, while the rounding errors occur when the generator's decimal amounts are rounded to on-chain integer amounts.

For simplicity and better understanding, consider that the generator takes an $M$-dimensional noise vector $\mathbf{X}$ as input, where each element consists of $n$ bits. The covert data is represented by an $M$-dimensional vector $\mathbf{CD}$, with each element consisting of $m$ bits. We denote the generator as $G(\cdot)$ and its inverse as $G^{-1}(\cdot)$. The data encoding process is illustrated in Algorithm \ref{alg: encoding}. The sender starts by sampling MSB and Padding for each covert data element $cd_i$, and concatenates them to form an $n$-bit noise element (lines 4-8). Using the noise vector composed of these noise elements, the sender obtains a decimal transaction amount vector $\mathbf{A}$ based on $G(\cdot)$ (lines 9-10). Since on-chain transaction amounts must be integers, decimal values are rounded to integer values $\hat{\mathbf{A}}$ (line 12). The sender then computes the recovered noise based on the rounded integer values (lines 13-14) and verifies whether the covert data bits in the recovered noise match those of the original noise (lines 15-26). The sender iteratively performs the embedding phase and the verification phase until a satisfying noise vector is obtained. Note that the infinite loop (line 2) can be prevented by randomizing MSB and decreasing the value of $m$. Directly setting the first $m$ bits of the noise as covert data without a randomizing MSB may incurs an infinite loop since only modifying the lower bits has less influence on the overall noise value. Increasing the value of $m$ allows for more covert data to be embedded in each transaction amount. However, this also leads to a longer time overhead in finding a satisfactory noise vector. The trade-off between $m$-value and the time overhead is explored in the experimental section.

\subsection{Data Decoding}

The main concept behind data decoding is to reverse the computation of the generator. Let the first layer has an $M$-dimensional input and an $H$-dimensional output, and the second layer has an $H$-dimensional input and an $N$-dimensional output. LeakyReLU has a slope of $\alpha$. Suppose the receiver obtains a transaction amount vector $\hat{\mathbf{A}} = (\hat{a}_1, \hat{a}_2, \cdots, \hat{a}_N) \in \mathbb{Z}^N_+$. The receiver can recover the covert data $\widehat{\mathbf{CD}} = (\widehat{cd}_1, \widehat{cd}_2, \cdots, \widehat{cd}_M) \in \{0, 1\}^{m \times M}$ in the following steps outlined in Fig. \ref{fig: decode}.
\begin{figure}[t]
\centerline{\includegraphics[width=0.45\textwidth]{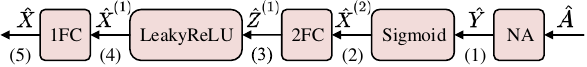}}
\captionsetup{justification=raggedright}
\caption{Data decoding. 1FC and 2FC refers to the first fully connected layer and the second fully connected layer. NA denotes normalization.}
\label{fig: decode}
\end{figure}

\begin{enumerate}[(1)]
  \item Normalize the transaction amount to obtain the output of Sigmoid. Given $\hat{\mathbf{A}}$, compute
  \begin{equation}
    \hat{\mathbf{Y}} = \frac{\hat{\mathbf{A}} - min(X)}{max(X) - min(X)},
  \end{equation}
  where $\hat{\mathbf{Y}} = (\hat{y}_1, \hat{y}_2, \cdots, \hat{y}_N) \in \mathbb{R}^N_+$.

  \item Calculate the output of the second fully connected layer. Given the output of the generator $\hat{\mathbf{Y}}$, compute
  \begin{equation}
    \hat{\mathbf{X}}^{(2)} = Logistic(\hat{\mathbf{Y}}),
  \end{equation}
  where $Logistic(\cdot)$ is the logistic function and $\hat{\mathbf{X}}^{(2)} = (\hat{x}^{(2)}_1, \hat{x}^{(2)}_2, \cdots, \hat{x}^{(2)}_N) \in \mathbb{R}^N$.

  \item Calculate the output of LeakyReLU. Given the output of the second fully connected layer $\hat{\mathbf{X}}^{(2)}$, compute
  \begin{equation}
    \hat{\mathbf{Z}}^{(1)} = \mathbf{W}_2^{-1}\hat{\mathbf{X}}^{(2)},
  \end{equation}
  where $\mathbf{W}_2^{-1}$ is the inverse matrix of the weight matrix of the second fully connected layer and $\hat{\mathbf{Z}}^{(1)} = (\hat{z}^{(1)}_1, \hat{z}^{(1)}_2, \cdots, \hat{z}^{(1)}_H) \in \mathbb{R}^H$. Note that $\mathbf{W}_2$ is a matrix with $N$ rows and $H$ columns , and $\mathbf{W}_2^{-1}$ exists only when $N = H$ and $\mathbf{W}_2$ is a full-rank matrix. We thus set $N = H$ in the model training process.

  \item Calculate the output of the first fully connected layer. Given the output of LeakyReLU $\hat{\mathbf{Z}}^{(1)} = (\hat{z}^{(1)}_1, \hat{z}^{(1)}_2, \cdots, \hat{z}^{(1)}_H) \in \mathbb{R}^H$, for each $\hat{z}^{(1)}_i$ where $i \in [1, 2, \cdots, H]$, compute
  \begin{equation}
    \hat{x}^{(1)}_i = \left\{
    \begin{aligned}
    & \hat{z}^{(1)}_i, \hat{z}^{(1)}_i \geq 0 \\
    & \frac{1}{\alpha}\hat{z}^{(1)}_i, \hat{z}^{(1)}_i < 0,
    \end{aligned}
    \right.
  \end{equation}
  and set $\hat{\mathbf{X}}^{(1)} = (\hat{x}^{(1)}_1, \hat{x}^{(1)}_2, \cdots, \hat{x}^{(1)}_H) \in \mathbb{R}^H$.

  \item Calculate the recovered noise $\hat{\mathbf{X}}$. Given $\hat{\mathbf{X}}^{(1)}$, compute
  \begin{equation}
    \hat{\mathbf{X}} = \mathbf{W}_1^{-1}\hat{\mathbf{X}}^{(1)},
  \end{equation}
  where $\mathbf{W}_1^{-1}$ is the inverse matrix of the weight matrix of the first fully connected layer and $\hat{\mathbf{X}} = (\hat{x}_1, \hat{x}_2, \cdots, \hat{x}_M) \in \mathbb{R}^M$. $\mathbf{W}_1$ is a matrix with $H$ rows and $M$ columns. We also set $H = M$ such that $\mathbf{W}_1^{-1}$ exists.

  \item Intercept and concatenate the recovered noise's covert data bits. Given $\hat{\mathbf{X}} = (\hat{x}_1, \hat{x}_2, \cdots, \hat{x}_M) \in \mathbb{R}^M$, compute
  \begin{equation}
    \widehat{\mathbf{CD}} = (\hat{x}_1[2:m+1], \cdots, \hat{x}_M[2:m+1]) \in \{0, 1\}^{m \times M},
  \end{equation}
  where $\hat{x}_i[2:m+1]$ refers to the bit string formed by concatenating the 2nd bit to the $(m+1)$th bit of $\hat{x}_i$.
\end{enumerate}

In summary, data decoding leverages the reversible nature of the generator to extract covert data from the transaction amount. To ensure the feasibility of the inverse calculation, both activation functions must be monotonic and continuous, and the generator's parameters $M$, $H$, and $N$ are set to be the same. The generator takes an M-dimensional noise vector as input and produces an M-dimensional transaction amount vector. Each element of the noise vector carries m-bit covert data. It is worth noting that there exists a computation error between the actual noise $\mathbf{X}$ and the recovered noise $\hat{\mathbf{X}}$. The value of $m$ is determined by the number of consecutive identical bits at the beginning of $\mathbf{X}$ and $\hat{\mathbf{X}}$. Therefore, a smaller computation error between $\mathbf{X}$ and $\hat{\mathbf{X}}$ allows for more bits of covert data to be encoded per noise/transaction amount. 
To enhance performance, we propose a counter-intuitive data preprocessing method and a custom activation function to increase the amount of covert data that can be embedded.

\section{Improved Covert Channel: CCR-GAN}

In this section, we propose CCR-GAN to improve R-GAN. The limited performance of R-GAN stems from the discrepancy between the noise recovered by the receiver and that sent by the sender, which arises from two errors~\cite{butora2023errorless}.
\begin{itemize}
  \item The first error is the precision error inherent to computers when performing decimal calculations, as there exists a built-in precision limit.
  \item The second error occurs when rounding transaction amounts. The generator of R-GAN generates a decimal value, while on-chain transaction amounts must be integers. The sender rounds the decimal number to an integer, which introduces errors.
\end{itemize}

The first error can be reduced by utilizing data types with higher precision. To reduce the second type of error, we initially introduce a counter-intuitive data preprocessing mechanism referred to as CIDP.

\subsection{Counter-intuitive Data Preprocessing}

Recall that the sender selects data with a higher occurrence probability as the training dataset. This choice allows the model to disregard extreme data and their influence on feature extraction, promoting generating data closely resembling normal data. 
Additionally, it encourages a more balanced and symmetrical distribution of the training data, which mitigates model overfitting and facilitates faster model convergence. However, we observe that this selection also leads to an increase in rounding errors. The rationale behind this observation is explained in the following.

For better clarity, we discuss the rounding error on a single transaction amount element instead of a transaction amount vector. To facilitate this discussion, we first propose the concept of the number of perfectly identical digits ($NPID$) between two real numbers ranging from $0$ to $1$.
\begin{definition}
  \textbf{(The number of perfectly identical digits, NPID).} Consider two decimal numbers, $0<a = a_0 \times 10^{0} + a_1 \times 10^{-1} + \cdots + a_m \times 10^{-m}<1$ and $0<b = b_0 \times 10^{0} + b_1 \times 10^{-1} + \cdots + b_n \times 10^{-n}<1$. We define $NPID(a, b)$ as the count of consecutive identical digits in $a$ and $b$ when counted from the integer digit backwards. Specifically, $NPID(a, b) = k+1$ if $a_k = b_k$ holds for all $i \in [0, 1, \cdots, k]$.
\end{definition}

Note that this definition represents $NPID$ as a decimal form, denoted by $NPID_{10}$. Alternatively, $NPID$ can also be expressed in binary form as $NPID_2$, indicating the number of consecutive unbroken identical bits. For clarity, we default $NPID$ to operate in decimal. In this context, $NPID$ represents the count of matching digits at each place value, starting from the integer digit, until a mismatched digit is encountered. For example, $NPID(1.81, 1.85) = 2$ and $NPID(0.1235, 0.1245) = 3$. We utilize $NPID$ as a measure of the rounding error, as it directly reflects the maximum amount of covert data recoverable by the receiver. A larger $NPID$ indicates a smaller rounding error. Next, we provide a justification for our observation that not dropping maximum and minimum extreme points of the training dataset, i.e., increasing the difference between the maximum and minimum values of the training dataset, results in an increased $NPID$.

\begin{theorem} \label{the: dp}
  Suppose the generator $G$ outputs $y$ ($0<y<1$), and the receiver inputs $\hat{y}$ to $G^{-1}$ during the data decoding process. When the minimum value of the training set is much smaller than the maximum value, for every ten-fold increase in the maximum value of the training dataset, then $NPID(y, \hat{y})$ is increased by $1$.
\end{theorem}

\begin{proof}
  We begin by using $y$ to represent $\hat{y}$. Let $X$ be the training dataset, $max(X)$ (an integer) denote the maximum element of $X$, and $min(X)$ (also an integer) denote the minimum element of $X$. The output $y$ is a normalized value. The sender denormalizes $y$ to obtain a decimal transaction amount $a$:
  \begin{equation} \label{equ: 1}
    a = y \times (max(X) - min(X)) + min(X).
  \end{equation}
  The sender then rounds $a$ to an integer:
  \begin{equation} \label{equ: 2}
    \hat{a} = [a].
  \end{equation}
  The receiver can only access $\hat{a}$ from the blockchain, which is normalized to serve as the input to $G^{-1}$:
  \begin{equation} \label{equ: 3}
    \hat{y} = \frac{\hat{a} - min(X)}{max(X) - min(X)}.
  \end{equation}
  Combining equations (\ref{equ: 1}), (\ref{equ: 2}), and (\ref{equ: 3}) gives
  \begin{equation}
    \begin{aligned}
    \hat{y} & = \frac{[y \times (max(X) - min(X)) + min(X)] - min(X)}{max(X) - min(X)} \\
    & = \frac{[y \times (max(X) - min(X))]}{max(X) - min(X)},
    \end{aligned}
  \end{equation}
  where $max(X) - min(X)$ is fixed for a given $X$. Since $min(X)$ is much smaller than $max(X)$, we have $max(X) - min(X) \approx max(X)$. Let $Q = max(X)$ for simplicity, then
  \begin{equation}
    \hat{y} = \frac{[yQ]}{Q}.
  \end{equation}
  Note that $Q$ is an integer and $y$ is a decimal in the range $0$ to $1$. Selecting transaction amounts with higher occurrence likelihood leads to a larger $Q$. 
  
  Now we show that for every ten-fold increase in $Q$, $NPID(y, \hat{y})$ increases by $1$. Let $y = y_1 \times 10^{-1} + \cdots + y_n \times 10^{-n}$, denoted as $\overline{0.y_1y_2 \cdots y_n}$. Let $Q = 10^m$. We have
  \begin{equation}
      \hat{y} = \frac{[\overline{y_1y_2 \cdots y_m.y_{m+1} \cdots y_n}]}{10^m},
  \end{equation}
  where $m \leq n$. Without loss of generality, we consider $y_{m+1} \leq 4$ ($y_{m+1}$ is rounded down when rounding). Then we have
  \begin{equation}
      \hat{y} = \frac{\overline{y_1y_2 \cdots y_m}}{10^m} = \overline{0.y_1y_2 \cdots y_m00 \cdots 0}.
  \end{equation}

  Without loss of generality, assume $y_{m+1} \neq 0$. In this case, $NPID(y, \hat{y}) = m + 1$. We complete the proof.
\end{proof}

Motivated by the theoretical justification mentioned above, we propose CIDP, which allows the sender to assign a larger value to $max(X) - min(X)$. In this approach, the sender includes all transaction amounts as part of the training dataset, without filtering out those with a higher likelihood of occurrence. By doing so, extremely large and small data points are not excluded, resulting in an increased $max(X) - min(X)$ and, consequently, an increased number of bits each transaction amount carries according to Theorem. \ref{the: dp}. However, not selecting transaction amounts incurs an extremely uneven distribution of normalized training data. About $83.98\%$ ($24.84\% + 41.05\% + 18.09\%$, see Fig.~\ref{fig: len_tx_amount} for detail) of the data falls between $10^{-8}$ and $10^{-5}$, while the entire range is from $0$ to $1$. This unevenness slows down the model's convergence speed, increases the risk of overfitting, and reduces the
diversity of generated data.

\subsection{ClipSigmoid}

CIDP can reduce the rounding error, while presenting a significant challenge in model convergence. In this section, we present ClipSigmoid as a solution to overcome this challenge. To improve readability, we first discuss the reasons behind this challenge and then provide the solution.

The convergence of R-GAN with CIDP is extremely difficult since the weight update gradient during backpropagation is too large relative to the input data. This leads to difficulties in achieving fine-grained variations in the model's output. In other words, each step taken by the model in the search for an optimal solution is too large to reach a better solution. The following content explains the occurrence of the long step.

Sigmoid has influence on the convergence rate of the model. Fig. \ref{fig: clip_justify} illustrates the weight update process related to Sigmoid. 
\begin{figure}[t]
\centerline{\includegraphics[width=0.4\textwidth]{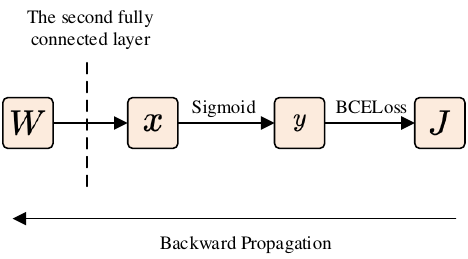}}
\captionsetup{justification=raggedright}
\caption{Sigmoid-related weight update process. $W$ represents any of the weight parameters in the model before the second fully connected layer, $x$ denotes the output of the second fully connected layer, $y$ is the output of Sigmoid, and $J$ refers to the loss function.}
\label{fig: clip_justify}
\end{figure}
Let $\sigma$ denote Sigmoid, which takes an input $x$ and produces an output $y$. The model updates its weights by:
\begin{equation}
  \begin{aligned}
  \Delta W & = \frac{\partial J}{\partial W} \\
  & = \frac{\partial J}{\partial y} \cdot \frac{\partial y}{\partial x} \cdot \frac{\partial x}{\partial W} \\
  & = \frac{\partial J}{\partial y} \cdot \sigma'(x) \cdot \frac{\partial x}{\partial W},
  \end{aligned}
\end{equation}
where $J$ represents the loss function, and $W$ is the model's weight. The derivative value of $\sigma$ is proportional to $\Delta W$. A larger value of $\sigma'$ results in a larger $\Delta W$, causing the model to take bigger steps in search of an appropriate solution. In R-GAN with CIDP, the weight update is particularly sensitive to the activation function. This is due to the fact that $83.98\%$ of the input data for R-GAN with CIDP falls within the range of $10^{-8}$ to $10^{-5}$, which is extremely tiny. Small weight changes (e.g., changes at the $10^{-2}$ level) have a considerable impact relative to the input data. Consequently, even small weight adjustments can significantly affect the model's outputs. Consider a generator initially producing numbers around $0.5$, while an ideal generator tends to generate numbers between $10^{-8}$ and $10^{-5}$. The initial generator tends to shift towards outputting numbers within the range of $10^{-8}$ to $10^{-5}$. However, due to the large step size, the model can easily generate data smaller than $10^{-8}$. Furthermore, when generating very small data, the generator mistakenly considers it to be of high quality. The discriminator categorizes these small data as real samples since the small data corresponds to the smallest value in the training dataset. For example, if the model outputs a $y<10^{-13}$, it is then denormalized to the transaction amount $[a] = [(max(X) - min(X)) \times y + min(X)] = min(X)$, i.e., the minimum value in the training dataset. The discriminator always assigns a true label to the minimum value since it is a genuine data point. As a result, the generator becomes unable to update the weights at a fine-grained level and tends to produce even smaller data. Eventually, the model is constrained to generating data that corresponds solely to the minimum value in the training dataset after denormalization.

An intuitive approach to address the challenges associated with CIDP is to modify the learning rate to mitigate the influence of weight updates on the model's outcomes. However, this approach is challenging in practice because determining when to decrease the learning rate and by how much is not straightforward. Furthermore, a very low learning rate can significantly prolong the training process.

To this end, we introduce a custom activation function called ClipSigmoid. Its main concept is to suppress the gradient during backpropagation when the model generates small data. This is achieved by setting the gradient of a specific segment of Sigmoid to zero. The definition of ClipSigmoid is as follows:
\begin{equation}
  ClipSigmoid(x) = \left\{
    \begin{aligned}
    & 1e^{-20}, Sigmoid(x) \leq 1e-20 \\
    & Sigmoid(x), Sigmoid(x) > 1e-20.
    \end{aligned}
    \right.
\end{equation}
\begin{figure}[t]
\centerline{\includegraphics[width=0.45\textwidth]{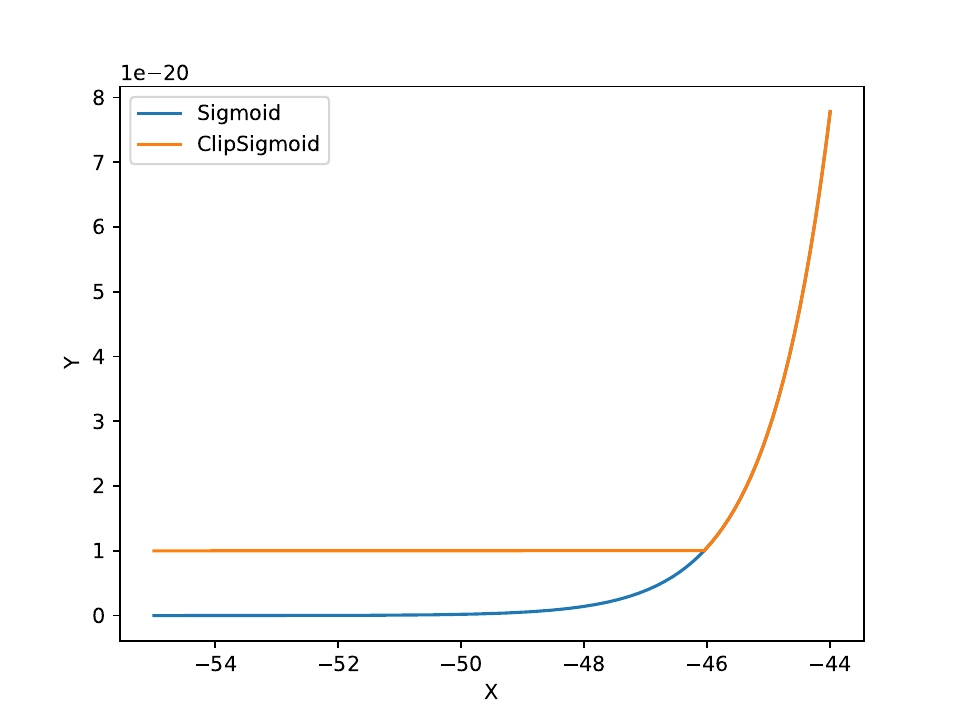}}
\captionsetup{justification=raggedright}
\caption{Different part of ClipSigmoid and Sigmoid.}
\label{fig: ClipSigmoid}
\end{figure}

Fig. \ref{fig: ClipSigmoid} illustrates the distinction between ClipSigmoid and Sigmoid. ClipSigmoid sets the lowest threshold for Sigmoid. When the value of Sigmoid falls below this threshold, ClipSigmoid enforces the value to this threshold. In our case, we set this threshold to $10^{-20}$, which is a hyperparameter established during the model training process. The threshold is set based on experience and is usually the square of the quotient of the minimum and maximum values in the dataset. Whenever the model generates a number below this threshold, the gradient of the weight update in the backpropagation step becomes zero. This effectively prevents the model from pursuing an incorrect solution by eliminating further exploration in that direction.

Note that the activation function is an integral part of R-GAN. It is thus necessary for ClipSigmoid to be reversible, allowing the receiver to compute the corresponding input $x$ from the model's output $y$. Although the zero-segment of ClipSigmoid is an irreversible straight line parallel to the $x$-axis, it does not impact the reversibility of R-GAN. The sender can intentionally train the generator to yield data outside the zero-segment to avoid the irreversible segment. The effectiveness of ClipSigmoid is presented in Section \ref{sec: effect_of_clipsigmoid}.

\section{T2C: Trade Capacity for Concealment}

In CCR-GAN, CIDP essentially reduces rounding error and increases capacity by increasing the magnitude of the dataset, which is defined as the logarithmic value of the largest element in the dataset with a base of 10. However, larger capacity typically results in lower concealment, as confirmed by Table~\ref{tab: CON}. The insight behind is embedding more covert data requires more transaction bits, reducing the number of bits available to resemble normal data. Technically, a larger magnitude means a more uneven dataset and a poorer quality of the trained models, which ultimately leads to a poorer concealment of the generated data. As a result, (CC)R-GAN trained on larger magnitude datasets may lead to very low concealment. For example, suppose that the communicating parties can accept up to 80\% of the recognized accuracy, i.e., the lower bound is 80\%. A particular transaction field supports embedding 40 bits of data in (CC)R-GAN and is recognized with only 90\% accuracy. The recognized accuracy exceeds the upper limit acceptable to the communicating parties. \emph{Hence, the question is, is it possible to balance capacity and concealment at a fine-grained level so that as much data as possible can be embedded within a given recognized accuracy?}

In this section, we propose T2C, a fine-grained approach to balance capacity and concealment. The core idea of T2C, which is inspired by CIDP, is to customize the magnitude of the dataset. When the dataset's magnitude is too large and leads to low concealment, the communicating parties can manually reduce the dataset's magnitude and increase the rounding error. In this way, the capacity is sacrificed to enhance the quality of the trainied model, which ultimately improves the concealment. Next, we show the technical details of T2C, which is implemented by decreasing- and recovering-magnitude algorithms.
\begin{algorithm}[htb]
  \caption{Decreasing magnitude.}
  \label{alg: decreasing_magnitude}
  \KwIn{Dataset $X = (x_1, x_2, \cdots, x_n)$. \\
  \quad \quad \quad Reduced magnitude $10^{\lambda}$.}
  \KwOut{Reduced dataset $DX = (dx_1, dx_2, \cdots, dx_n)$.}

  Initialize $DX = (dx_1, dx_2, \cdots, dx_n)$\;
  \For{$dx_i \in DX$}{
      $dx_i = \frac{x_i}{10^\lambda}$\;
  }
  \Return{$DX$}
\end{algorithm}

Algorithm~\ref{alg: decreasing_magnitude} demonstrates the process of decreasing magnitude. Its essence is to reduce each element in the dataset by the same magnitude. In Algorithm~\ref{alg: decreasing_magnitude}, we assume that the reduced magnitude is $10^\lambda$, which is consistent with the expression of $Q$ in Theorem~\ref{the: dp}. This makes it more intuitive for the reader to understand the justification of Theorem~\ref{the: dp} for T2C. In practice, the communicating parties can use 2 as the base number and use $NPID_2$ to measure the rounding error in order to sacrifice the magnitude and increase concealment on a bit-by-bit basis. After decreasing the magnitude, the normalized dataset used for training is more homogeneous compared to before decreasing the magnitude. Therefore, the trained model is of higher quality, and the generated data also possesses a higher degree of concealment.

\begin{algorithm}[htb]
  \caption{Recovering magnitude.}
  \label{alg: recovering_magnitude}
  \KwIn{Dataset $X = (x_1, x_2, \cdots, x_n)$.\\
  \quad \quad \quad Reduced magnitude $10^{\lambda}$.\\
  \quad \quad \quad Generated data $\mathbf{A} = (a_1, a_2, \cdots, a_n)$.}
  \KwOut{On-chain data $\hat{\mathbf{A}} = (\hat{a}_1, \hat{a}_2, \cdots, \hat{a}_n)$.}

  Initialize $\hat{\mathbf{A}} = (\hat{a}_1, \hat{a}_2, \cdots, \hat{a}_n)$\;
  Initialize a Key-Value dictionary $D = \{\}$\;
  // \textbf{Count the frequency of the last $\lambda$ digits of each element $x_i$ in the original dataset $X$}\;
  \For{$x_i \in X$}{
    Suppose the last $\lambda$ digits of $x_i$ are $t_i$\;
    \If{$t_i$ already exists in the Key of the dictionary $D$}{
      The corresponding Value $D[t_i] = D[t_i]+1$\;
    }
    \Else($t_i$ never appears in the Key of dictionary $D$){
      Create a new Key $D[t_i]$ and assign $D[t_i]=1$\;
    }
  }
  \For{$a_i \in \mathbf{A}$}{
    Sample from $D$ to get a number $t$ with length $\lambda$\;
    Set $\hat{a}_i = [a_i] * 10^\lambda + t$\;
  }
  \Return{$\hat{\mathbf{A}}$}
\end{algorithm}

Since the magnitude of the dataset used for training is reduced, the model will only output data with smaller magnitude. For example, if the original dataset $X$ is of magnitude $10^{14}-10^{18}$, the dataset used for training becomes of magnitude $10^{12}-10^{16}$ after the magnitude is reduced by $10^2$. The trained model will also output data of magnitude $10^{12}-10^{16}$. The sender needs to recover the data to $10^{14}-10^{18}$ again. Otherwise, the obviously smaller data (between $10^{12}-10^{14}$) can be easily detected by the adversary. Algorithm~\ref{alg: decreasing_magnitude} describes the process of revovering magnitude. For a given original dataset $X$, the sender needs to count the last $\lambda$ digits of all elements in $X$ to form a distribution $D$. For an arbitrarily generated data $a_i$, the sender first rounds $a_i$ (to get $[a_i]$), then samples a digit with $\lambda$-length from the distribution $D$ and concatenates it to $[a_i]$ to obtain a recovering magnitude number $\hat{a_i}$. This $\hat{a_i}$ is the final on-chain data. Note that T2C does not have an impact on the accuracy of the receiver's obtaining on-chain data. This is because the reduced magnitude data generated by the model is not lost, but instead $\lambda$-length redundant digits are added. The receiver only needs to negotiate $\lambda$ with the sender and truncate the last $\lambda$ digits of the on-chain data to get the data generated by the trained model.

In summary, the sender trains the model by first decreasing the magnitude of the dataset, thus increasing the dataset and the quality of the trained model (i.e., increasing the concealment) at the cost of increasing the rounding error (i.e., decreasing the capacity). At this point, the model generates reduced magnitude data. The sender then recovers the magnitude by sampling from the last $\lambda$ digits of the original data. Moreover, when $\lambda$ is negative, T2C can increase the capacity at the expense of concealment. At this point, Algorithm~\ref{alg: decreasing_magnitude} functions to increase the magnitude. The model generates data with larger magnitude. The sender then needs to calculate $\hat{a_i}=[a_i \times 10^\lambda]$ to recover the magnitude.

\section{Experiments}


In this section, we first evaluate R-GAN and CCR-GAN using the Bitcoin transaction amount as an example. Then, we apply the proposed method to the Bitcoin transaction fee field as well as to Ethereum to verify the scalability. We further evaluate the capacity and concealment aspects of proposed schemes. Finally, we compare R-GAN and CCR-GAN in terms of capacity and concealment with baselines. This section aims to address the following research questions:
\begin{itemize}
  \item \textbf{RQ1:} How much data can R-GAN and CCR-GAN embed in Bitcoin transaction amounts?
  \item \textbf{RQ2:} Can R-GAN and CCR-GAN be applied to other transaction fields and can they be extended to other blockchains?
  \item \textbf{RQ3:} Since CCR-GAN and T2C can reduce the rounding error by increasing the magnitude of the dataset, is it possible to increase the magnitude infinitely to embed more data?
  \item \textbf{RQ4:} How about the concealment of R-GAN and CCR-GAN?
  \item \textbf{RQ5:} How about the effect of T2C?
  \item \textbf{RQ6:} How to demonstrate the effectiveness of ClipSigmoid?
  \item \textbf{RQ7:} What is the performance of R-GAN and CCR-GAN compared to existing schemes?
\end{itemize}

\subsection{Setup}

Both R-GAN and CCR-GAN are implemented using Python and PyTorch 1.13.1~\cite{pytorch}, and trained on machines with an Intel(R) Core(TM) i5-8265U CPU @ 1.60GHz 1.80 GHz and 8.00 GB RAM. The training process is very fast and can be completed in a few minutes on the CPU.

\subsubsection{Dataset}

We collect a dataset including 84,515 output amounts from Bitcoin transactions downloaded from block $727215$ to block $727239$ in the Bitcoin mainnet. These amounts are measured in satoshi and are represented as integers ranging from $296$ to $2,874,993,345,277$. In CCR-GAN, we discard amounts that cannot form a complete batch of input data based on the batch size and the input dimension. In the case of R-GAN, we include amounts between $10^5$ and $10^8 - 1$ as the training set, and also discard redundant amounts.

\subsubsection{Baselines}

We compare our schemes and four baselines that use fields other than the output amount as the embedding field to evaluate the channel expansion capacity: \textbf{(1) HC-CDE}~\cite{cao2020chain} encodes covert data using the computational relationship between transaction addresses. \textbf{(2) DSA}~\cite{fionov2019exploring} includes schemes that replace random factors in the signing process with covert data. \textbf{(3) Un-UTXO} \cite{gregoriadis2022analysis} encodes covert data as an output address. \textbf{(4) DLchain}~\cite{tian2020dlchain} utilizes the private key as the carrier of covert data. These baselines are widely recognized for high concealment, which makes them suitable for comparison. In addition, we compare Bitcoin amount with (CC)R-GAN to three baselines that also utilize Bitcoin amount as the embedding field, including \textbf{CCMBBT}~\cite{luo2021novel}, \textbf{STCBC}~\cite{zhang2024blockchain}, and \textbf{AMASC}~\cite{tian2024amount}.

\subsubsection{Performance metrics}

We adopt the following three metrics as performance metrics. 

\begin{itemize}
  \item \textbf{Absolute capacity (AC)} indicates the amount of data that can be embedded in each expansion field. It is defined as:
  \begin{equation}
    \mathsf{AC} = \frac{n}{N_f} \ \  bit,
  \end{equation}
  where $N_f$ is the number of expansion fields, and $n$ is the number of bits of covert data carried by these expansion fields. 

  \item \textbf{Capacity expansion rate (CER)} refers to the capacity by which the expansion covert channel enhances the original covert channel. It is computed by:
  \begin{equation} \label{equ: CER}
    \mathsf{CER} = \frac{\mathsf{AC}_{ecc}}{\mathsf{AC}_{occ}} \times 100\%,
  \end{equation}
  where $\mathsf{AC}_{ecc}$ represents $\mathsf{AC}$ of the expansion covert channel, and $\mathsf{AC}_{occ}$ represents $\mathsf{AC}$ of the original covert channel. 

  \item \textbf{Concealment.} Concealment refers to the ability of embedding/expansion fields to remain undetected. We evaluate the concealment using the accuracy, precision, recall, and F1\_score of CTR model~\cite{wang2022practical}, which is a steganalysis model used to detect the presence of covert data within a transaction field. A concealment score closer to $0.5$ indicates a higher level of concealment.
\end{itemize}

\subsubsection{Parameter setting} 

The batch size is set to 10, and the input/hidden/output dimensions are set to 64. LeakyReLU are with $\alpha = 0.3$. The learning rate of both R-GAN and CCR-GAN is set to 0.001, while the learning rate of CCR-GAN decreases to 0.9 times the original value after every two epochs. CCR-GAN stops training when the loss does not decrease in 10 consecutive epochs, while R-GAN stops training when that does not decrease in 5 consecutive epochs. The amount values are normalized using the maximum-minimum normalization method. All parameters, including weights, biases, and input/output data, are of type float64 to minimize the computational precision error.

\subsection{Capacity on Bitcoin Amount} \label{sec: self_exp}



\textbf{We begin by evaluating $\bm{\mathsf{AC}}$ in terms of Bitcoin transaction amount, with the first step being to determine a rough value for $\bm{m}$.} We calculate $m = \text{min}(NPID_2(\hat{x}_i, x_i)) - 1$, where $i \in \{1, 2, \cdots, M\}$ where $NPID_2$ denotes the binary form of $NPID$ and $M$ denotes the input dimension. The first $m+1$ bits of each element in $\hat{\mathbf{X}}$ and the corresponding bits in $\mathbf{X}$ should be identical to ensure accurate recovery of the $m$-bit covert data by the receiver. When obtaining a trained model, we need to measure the approximate $m$ in the following manner.

Given a trained model, we input a completely random noise $\mathbf{X}$, calculate the reduced noise $\hat{\mathbf{X}}$ obtained by the receiver, and compute the minimum recovery bit as $min(NPID_2(\hat{x}_i, x_i))$. To determine the appropriate range of $m$, we repeat the above steps 10,000 times. Experimental results are presented in Fig. \ref{fig: recovery_bit}. We focus on the highest number of recovered bits. For R-GAN, 1 out of 10,000 experiments is able to recover 12 bits of noise. For CCR-GAN, 6 out of 10,000 experiments can recover 25 bits of noise. The results indicate that R-GAN and CCR-GAN have the potential to allow the receiver to recover approximately 11 and 25 bits of covert data, respectively. Meanwhile, the sender must consider the computational resources and time required during the data encoding's verification phase. 
\begin{figure}[t]
  \centering
  \subfigure[Recovery bits for R-GAN.]{
      \begin{minipage}[t]{0.45\linewidth}
          \centering
      \includegraphics[width=1.7in]{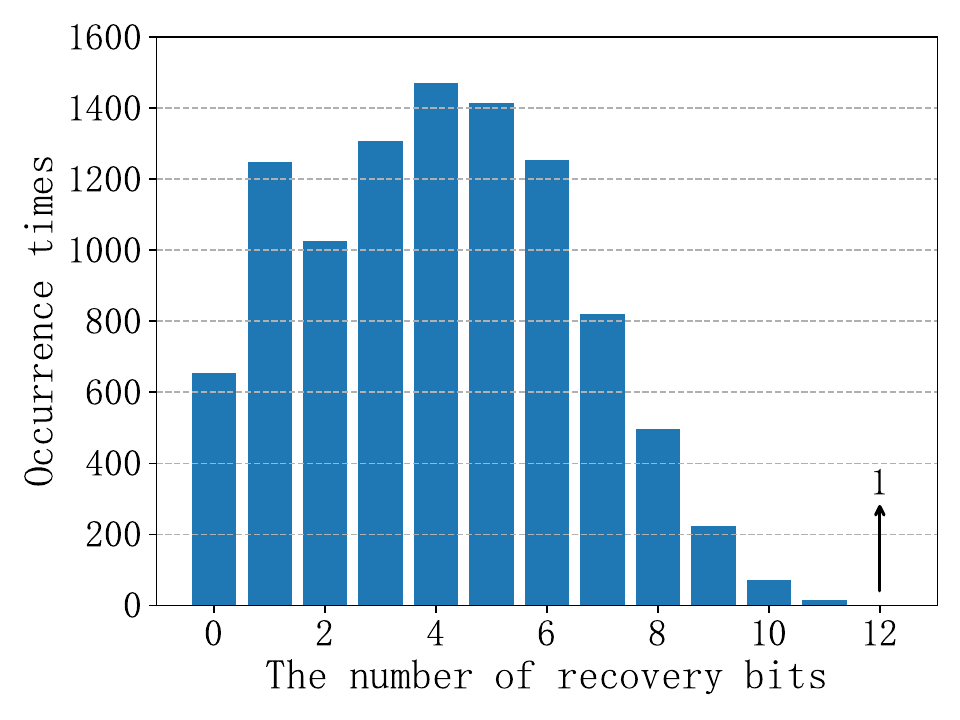}
      \label{fig: R_GAN_recovery_bit}
  \end{minipage}
  }
  \subfigure[Recovery bits for CCR-GAN.]{
      \begin{minipage}[t]{0.45\linewidth}
          \centering
      \includegraphics[width=1.7in]{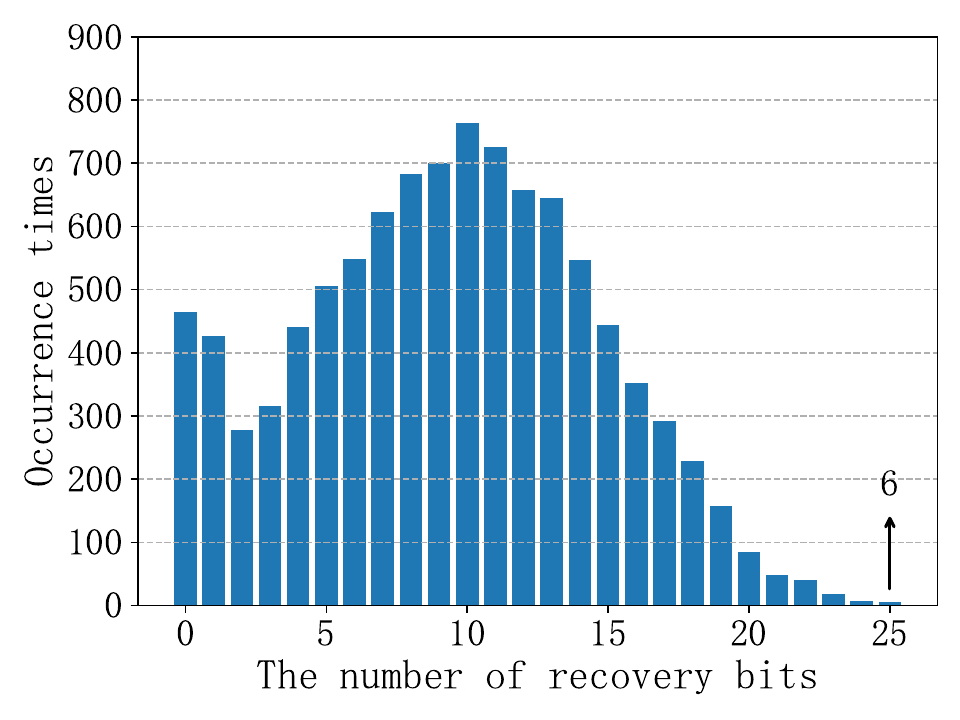}
      \label{fig: CCR_GAN_recovery_bit}
      \end{minipage}
  }
  \caption{Experiments on recovery bits (including 1-bit MSB).}
  \label{fig: recovery_bit}
\end{figure}
\begin{figure}[t]
  \centering
  \subfigure[Test $m$ for R-GAN.]{
      \begin{minipage}[t]{0.45\linewidth}
          \centering
      \includegraphics[width=1.7in]{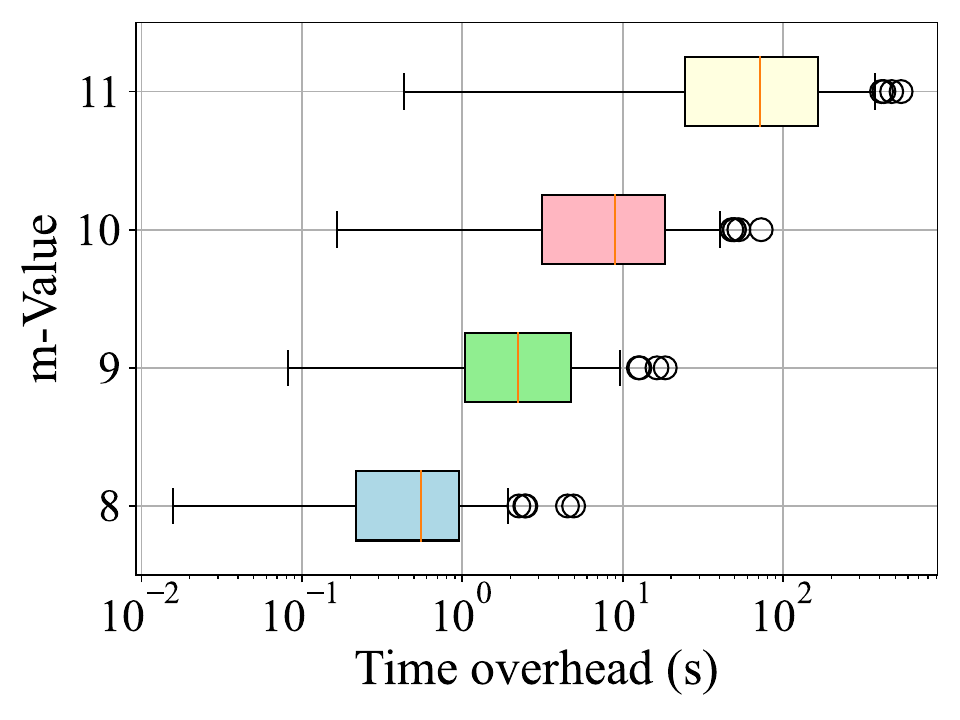}
      \label{fig: R_GAN_m}
  \end{minipage}
  }
  \subfigure[Test $m$ for CCR-GAN.]{
      \begin{minipage}[t]{0.45\linewidth}
          \centering
      \includegraphics[width=1.7in]{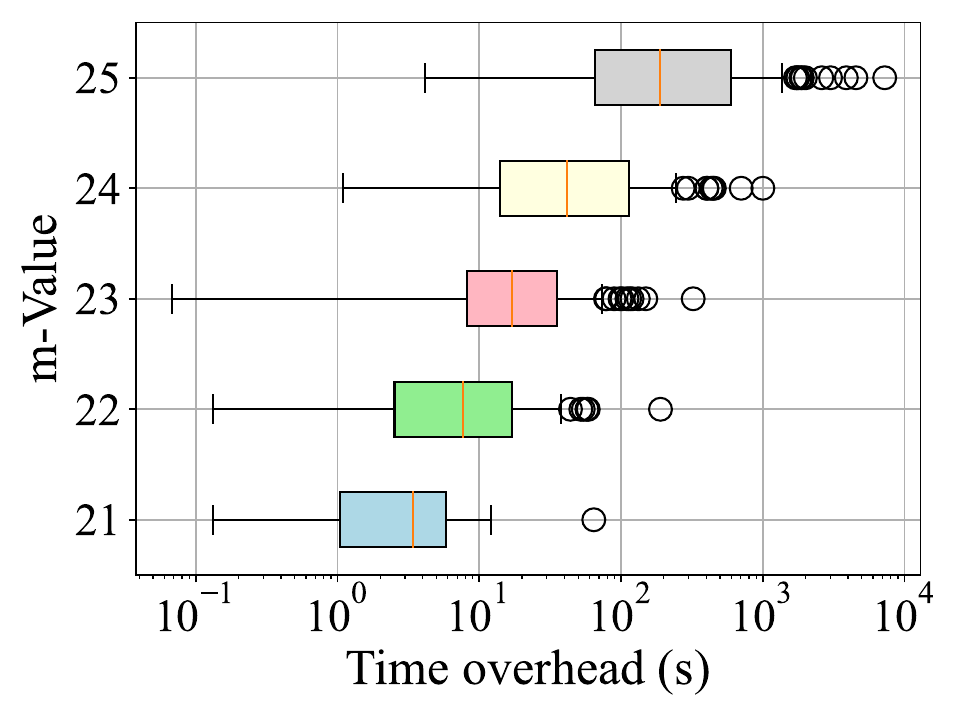}
      \label{fig: CCR_GAN_m}
      \end{minipage}
  }
  \caption{Time overhead to find a satisfying noise vector.}
  \label{fig: time_m}
\end{figure}

\begin{table*}[htb]
  \scriptsize
  \captionsetup{justification=raggedright}
  \caption{Time overhead required to find a satisfying noise and noise vector for R-GAN and CCR-GAN.}
  \begin{center}
  \begin{tabular}{c c c c c c c c c c c c}
  \hline
  \multirow{2}{*}{\textbf{Field}} & \multirow{2}{*}{\textbf{Scheme}} & \multirow{2}{*}{\textbf{Dataset magnitude}} & \multirow{2}{*}{$\bm{m}$} & \multicolumn{4}{c}{\textbf{Time overhead for a vector (s)}} & \multicolumn{4}{c}{\textbf{Time overhead for an amount (s)}} \\
  \cline{5-12}
   & & & & 1st Quartile & 2nd Quartile & 3rd Quartile & Average & 1st Quartile & 2nd Quartile & 3rd Quartile & Average \\
  \hline
  \multirow{8}{*}{\shortstack{Bitcoin \\ amount}} & \multirow{4}{*}{R-GAN} & \multirow{4}{*}{$10^4\sim10^7$} & $8$  & $0.216 $ & $0.552 $ & $0.959  $ & $0.736  $ & $0.003$ & $0.003$ & $0.015$ & $0.011$ \\
                                                  &                        &                                 & $9$  & $1.033 $ & $2.215 $ & $4.748  $ & $3.405  $ & $0.016$ & $0.035$ & $0.074$ & $0.053$ \\
                                                  &                        &                                 & $10$ & $3.160 $ & $8.934 $ & $18.430 $ & $13.215 $ & $0.049$ & $0.140$ & $0.288$ & $0.206$ \\
                                                  &                        &                                 & $11$ & $24.365$ & $72.127$ & $\bm{165.125}$ & $119.866$ & $0.381$ & $1.127$ & $\bm{2.580}$ & $1.873$ \\
  \cline{2-12}
                                  & \multirow{5}{*}{CCR-GAN} & \multirow{5}{*}{$10^2\sim10^{13}$} & $21$ & $1.038 $ & $3.419  $ & $5.852  $ & $4.385  $ & $0.016$ & $0.053$ & $0.091$ & $0.068$ \\
                                  &                          &                                    & $22$ & $2.527 $ & $7.677  $ & $17.155 $ & $13.766 $ & $0.039$ & $0.120$ & $0.268$ & $0.215$ \\
                                  &                          &                                    & $23$ & $8.170 $ & $17.055 $ & $35.379 $ & $31.931 $ & $0.128$ & $0.266$ & $0.553$ & $0.499$ \\
                                  &                          &                                    & $24$ & $14.022$ & $41.502 $ & $\bm{113.695}$ & $100.168$ & $0.219$ & $0.648$ & $\bm{1.776}$ & $1.565$ \\
                                  &                          &                                    & $25$ & $65.929$ & $188.369$ & $594.831$ & $584.567$ & $1.030$ & $2.943$ & $9.294$ & $9.134$ \\
  \hline
  \multirow{6}{*}{\shortstack{Bitcoin \\ fee}} & \multirow{2}{*}{R-GAN} & \multirow{2}{*}{$10^2\sim10^4$} & $1$ & $0.135  $ & $0.451  $ & $0.893    $ & $0.943    $ & $0.002$ & $0.007$ & $0.014$ & $0.015$ \\
                                               &                        &                                 & $2$ & $30.306 $ & $74.233 $ & $\bm{199.201}  $ & $183.078  $ & $0.474$ & $1.160$ & $\bm{3.113}$ & $2.861$ \\
  \cline{2-12}
                                               & \multirow{4}{*}{CCR-GAN} & \multirow{4}{*}{$10^2\sim10^7$} & $1$ & $0.224 $ & $0.524  $ & $1.737  $ & $1.777  $ & $0.004$ & $0.008$ & $0.027$ & $0.028$ \\
                                               &                          &                                 & $2$ & $1.669 $ & $4.670  $ & $9.271  $ & $6.251  $ & $0.026$ & $0.073$ & $0.145$ & $0.098$ \\
                                               &                          &                                 & $3$ & $4.126 $ & $11.812 $ & $28.032 $ & $26.553 $ & $0.064$ & $0.185$ & $0.438$ & $0.415$ \\
                                               &                          &                                 & $4$ & $25.869$ & $65.478 $ & $\bm{156.477}$ & $136.903$ & $0.404$ & $1.023$ & $\bm{2.445}$ & $2.139$ \\
  \hline
  \multirow{8}{*}{\shortstack{Ethereum \\ amount}} & \multirow{4}{*}{R-GAN} & \multirow{4}{*}{$10^{14}\sim10^{18}$} & $38$ & $0.094 $ & $0.196 $ & $0.351  $ & $0.259  $ & $0.001$ & $0.003$ & $0.005$ & $0.004$ \\
                                                   &                        &                                       & $39$ & $0.906 $ & $2.139 $ & $4.988  $ & $3.457  $ & $0.014$ & $0.033$ & $0.078$ & $0.054$ \\
                                                   &                        &                                       & $40$ & $5.493 $ & $11.997 $ & $19.949 $ & $16.385 $ & $0.086$ & $0.187$ & $0.312$ & $0.256$ \\
                                                   &                        &                                       & $41$ & $18.853$ & $46.033$ & $\bm{89.295}$ & $72.827$ & $0.295$ & $0.719$ & $\bm{1.395}$ & $1.138$ \\
  \cline{2-12}
                                                   & \multirow{4}{*}{CCR-GAN} & \multirow{4}{*}{$10^3\sim10^{23}$} & $37$ & $0.154  $ & $0.271  $ & $0.531   $ & $0.427   $ & $0.002$ & $0.004$ & $0.008$ & $0.007$ \\
                                                   &                          &                                    & $38$ & $0.452  $ & $1.110  $ & $2.209   $ & $1.450   $ & $0.007$ & $0.017$ & $0.035$ & $0.023$ \\
                                                   &                          &                                    & $39$ & $2.818  $ & $5.669  $ & $11.687  $ & $7.874   $ & $0.044$ & $0.089$ & $0.183$ & $0.123$ \\
                                                   &                          &                                    & $40$ & $47.736 $ & $93.226 $ & $\bm{163.406} $ & $113.912 $ & $0.745$ & $1.457$ & $\bm{2.553}$ & $1.780$ \\
  \hline
  \multirow{8}{*}{\shortstack{Ethereum \\ fee}} & \multirow{4}{*}{R-GAN} & \multirow{4}{*}{$10^9\sim10^{11}$} & $19$  & $0.167 $ & $0.382  $ & $0.887  $ & $0.687  $ & $0.003$ & $0.006$ & $0.014$ & $0.011$ \\
                                                &                        &                                    & $20$  & $0.583 $ & $1.287  $ & $2.671  $ & $2.029  $ & $0.009$ & $0.020$ & $0.042$ & $0.032$ \\
                                                &                        &                                    & $21$  & $2.505 $ & $6.327  $ & $12.003 $ & $10.380 $ & $0.039$ & $0.099$ & $0.188$ & $0.162$ \\
                                                &                        &                                    & $22$  & $11.803$ & $32.450 $ & $\bm{71.229}$ & $48.917$ & $0.184$ & $0.507$ & $\bm{1.113}$ & $0.764$ \\
  \cline{2-12}
                                & \multirow{4}{*}{CCR-GAN} & \multirow{4}{*}{$10^9\sim10^{13}$} & $22$ & $0.408 $ & $0.740  $ & $1.467  $ & $1.018  $ & $0.006$ & $0.012$ & $0.023$ & $0.016$ \\
                                &                          &                                    & $23$ & $0.984 $ & $2.338  $ & $4.396  $ & $4.096  $ & $0.015$ & $0.037$ & $0.069$ & $0.064$ \\
                                &                          &                                    & $24$ & $5.095 $ & $16.498 $ & $33.353 $ & $27.780 $ & $0.080$ & $0.258$ & $0.521$ & $0.434$ \\
                                &                          &                                    & $25$ & $27.642$ & $53.189 $ & $\bm{136.863}$ & $97.454$ & $0.432$ & $0.831$ & $\bm{2.138}$ & $1.523$ \\
  \hline
  \end{tabular}
  \label{tab: time_m}
  \end{center}
  \end{table*}

\textbf{The second step is to determine the accurate value of $\bm{m}$ for R-GAN and CCR-GAN based on the time consumption.} When embedding certain covert data, the variability of the noise is reduced compared to the experiments in the last step. The covert data is embedded in bits $2$ to $m+1$ of the noise, meaning these $m$ bits are fixed and cannot be randomized. Furthermore, these $m$ bits occupy the more significant positions which have a greater influence on the overall noise value. When significant bits are fixed, it may take more time (or may not even be possible) to find a satisfactory noise vector. We select ``The Little Prince" as transmitted data to align with CCMBBT, and the covert data is encrypted using AES-CBC.

In the case of R-GAN, we assess the time overhead of finding 100 noise vectors for different covert data where the values of $m$ ranging from 8 to 11. Fig. \ref{fig: R_GAN_m} illustrates the results. The time overhead of finding a satisfactory noise vector increases exponentially as $m$ increases. When $m$ reaches 11, it takes approximately 100 seconds on average to find a suitable noise vector (including 64 noises). 
For CCR-GAN, we adopt the same evaluation criteria as CCR-GAN, while $m$ varies from 21 to 25. Results are presented in Fig. \ref{fig: CCR_GAN_m}. When $m$ is within the range of 21 to 23, the time overhead remains relatively consistent. When $m$ exceeds 24, the time overhead begins to increase significantly. The majority of boxes in Fig. \ref{fig: time_m} are positioned close to the maximum point, indicating that most results fall within the range near the maximum.

Table \ref{tab: time_m} provides more detailed time consumption for finding a satisfying noise vector and noise. All the means are higher than the medians, verifying that a majority of the results are close to the maximum. If the time consumption of the upper quartile is deemed acceptable, we consider the scheme to be feasible. For R-GAN with $m=11$, the sender typically spends 165.125 and 2.580 seconds to find a satisfying noise vector and noise, which is deemed acceptable. In the case of CCR-GAN, the sender can accept the time overhead at $m=24$. When $m=25$, it may take 594.831 seconds and 9.134 seconds to find a noise vector and noise, which is unacceptable.


\begin{tcolorbox}[colframe=black, colback=gray!10, boxrule=1pt,width=\linewidth]
  \textbf{Answer to RQ1:} R-GAN and CCR-GAN allow for 11-bit and 24-bit covert data to be embedded in each Bitcoin transaction amount and are capable of generating a data-carrying amount in less than 3 seconds.
\end{tcolorbox}

\subsection{Scalability}

\textbf{Then, we verify the scalability. Specifically, we apply R-GAN and CCR-GAN to the Bitcoin fee, the Ethereum amount, and the Ethereum fee, and use the same experimental approach to determine the value of $\bm{m}$.} We chose the amount and fee for evaluation because the values of these fields are completely specified by the sender. Futhermore, most schemes do not use the amount and fee as the embedding field, which means that R-GAN and CCR-GAN can effectively boost the capacity of most existing schemes. We capture 80,000 amounts and 80,000 fees from the real Bitcoin and Ethereum networks to construct the dataset, and Table~\ref{tab: time_m} shows the experimental results.

\textbf{For Bitcoin fee}, R-GAN and CCR-GAN only support embedding 2-bit and 4-bit data. R-GAN's dataset is in the range of $100$ to $9999$, and the amount of data that can be embedded is consequently small according to Theorem~\ref{the: dp}. CCR-GAN's dataset has a maximum of $10^7$ magnitude, and is able to embed a slightly larger amount of data.

\textbf{For Ethereum amount}, R-GAN and CCR-GAN are capable of embedding 41 and 40 bits of data. Despite the significant increase in the magnitude of the CCR-GAN dataset, the amount of data that can be embedded does not improve. This does not mean that CCR-GAN fails, but rather because the precision error of the computer determines the upper limit of the amount of data embedded. In our experiments, we use the highest precision floating-point type in pytorch, float64, which follows the IEEE 754 double-precision floating-point standard and is accurate to approximately $17$ decimal places. Thus, when the magnitude of the largest value in the dataset reaches $10^{17}$, continuing to increase the magnitude will not increase the amount of data embedded, as the computer's precision error limits the accuracy of the data recovered by the receiver. To verify the conclusion, we divide all the data in the dataset by a certain power of $10$ to reduce the magnitude and train the R-GAN with the reduced magnitude data, and then determine the value of $m$. The results are shown in Table~\ref{tab: eth_amount_improve}.
\begin{table*}[htb]
  \scriptsize
  \captionsetup{justification=raggedright}
  \caption{Time overhead required to find a satisfying noise and noise vector for R-GAN mode-Ethereum amount when reducing magnitude at different $m$.}
  \begin{center}
  \begin{tabular}{c c c c c c c c c c c}
  \hline
  \multirow{2}{*}{\textbf{Reduced magnitude}} & \multirow{2}{*}{\textbf{Dataset magnitude}} & \multirow{2}{*}{$\bm{m}$} & \multicolumn{4}{c}{\textbf{Time overhead for a vector (s)}} & \multicolumn{4}{c}{\textbf{Time overhead for an amount (s)}} \\
  \cline{4-11}
   & & & 1st Quartile & 2nd Quartile & 3rd Quartile & Average & 1st Quartile & 2nd Quartile & 3rd Quartile & Average \\
  \hline
  \multirow{4}{*}{$10^1$} & \multirow{4}{*}{$10^{13}\sim10^{17}$} & $37$ & $0.417  $ & $1.103  $ & $1.931  $ & $1.441  $ & $0.007$ & $0.017$ & $0.030$ & $0.023$ \\
                          &                                       & $38$ & $0.491  $ & $1.237  $ & $2.212  $ & $1.641  $ & $0.008$ & $0.019$ & $0.035$ & $0.026$ \\
                          &                                       & $39$ & $10.133 $ & $24.769 $ & $47.629 $ & $35.372 $ & $0.158$ & $0.387$ & $0.744$ & $0.553$ \\
                          &                                       & $40$ & $23.043 $ & $62.644$ & $\bm{156.099}$ & $101.800$ & $0.360$ & $0.979$ & $\bm{2.439}$ & $1.591$ \\
  \hline
  \multirow{4}{*}{$10^2$} & \multirow{4}{*}{$10^{12}\sim10^{16}$} & $35$ & $1.275  $ & $2.774   $ & $7.195   $ & $4.516    $ & $0.020$ & $0.043$ & $0.112$ & $0.071$ \\
                          &                                       & $36$ & $1.014  $ & $3.907   $ & $8.287   $ & $5.822    $ & $0.016$ & $0.061$ & $0.129$ & $0.091$ \\
                          &                                       & $37$ & $5.857  $ & $13.144  $ & $28.562  $ & $20.760   $ & $0.092$ & $0.205$ & $0.446$ & $0.324$ \\
                          &                                       & $38$ & $27.076 $ & $79.351 $ & $\bm{197.818} $ & $131.941 $ & $0.423$ & $1.240$ & $\bm{3.091}$ & $2.062$ \\
  \hline
  \multirow{4}{*}{$10^3$} & \multirow{4}{*}{$10^{11}\sim10^{15}$} & $32$ & $1.878   $ & $6.480   $ & $11.285  $ & $8.221   $ & $0.029$ & $0.101$ & $0.176$ & $0.128$ \\
                          &                                       & $33$ & $4.160   $ & $13.286  $ & $27.916  $ & $19.867  $ & $0.065$ & $0.208$ & $0.436$ & $0.310$ \\
                          &                                       & $34$ & $10.483  $ & $39.158  $ & $65.814  $ & $52.042  $ & $0.164$ & $0.612$ & $1.028$ & $0.813$ \\
                          &                                       & $35$ & $71.394  $ & $174.085 $ & $\bm{354.424} $ & $248.024 $ & $1.116$ & $2.720$ & $\bm{5.538}$ & $3.875$ \\
  \hline                          
  \end{tabular}
  \label{tab: eth_amount_improve}
  \end{center}
\end{table*}
When the magnitude of the dataset is reduced by $10^1$, the magnitude of the dataset's maximum value is $10^{17}$. According to Theorem~\ref{the: dp}, at this point the rounding step can accurately retain 17 digits after the decimal point, and this precision is consistent with the maximum precision of pytorch's float64-types numbers. Therefore, the amount of data that can be embedded theoretically will not decrease. The experimental results in Table~\ref{tab: eth_amount_improve} show that 40 bits of data can still be embedded. The experimental results are consistent with Theorem~\ref{the: dp}. 
As the magnitude of the largest data in the dataset continues to decrease to $10^{16}$, even though the computer is able to accurately compute numbers to 17 decimal places, the rounding step is only able to ensure that the first 16 places are accurate, thus reducing the amount of data that can be embedded to 38 bits.
Continuing to reduce the magnitude of the dataset, the precision of rounding gradually decreases, and the amount of data that can be embedded also gradually decreases, which also confirms our conclusion.

\textbf{The Ethereum fee} can also be applied to R-GAN and CCR-GAN, and CCR-GAN is able to embed more data.R-GAN is able to embed 22 bits of data in the Ethereum fee, and CCR-GAN is able to embed 25 bits of data in the Ethernet fee.

\begin{tcolorbox}[colframe=black, colback=gray!10, boxrule=1pt,width=\linewidth]
  \textbf{Answer to RQ2:} R-GAN and CCR-GAN are scalable and able to be applied in numerical transaction fields of various blockchains.
\end{tcolorbox}

A finding that further supports Theorem~\ref{the: dp} is that the amount of embedded data ($m$-value) is determined by the magnitude of the dataset's maximum data. For Bitcoin amount and Ethereum fee, the larger the dataset magnitude, the larger the amount of data that can be embedded, and when the magnitude is the same (CCR-GAN for Bitcoin amount and CCR-GAN for Ethereum fee), the amount of data that can be embedded is also the same (24 bits). One exception is that CCR-GAN for Bitcoin fee and R-GAN for Bitcoin amount have the same maximum magnitude ($10^7$), but the data embedding for Bitcoin fee is very small. This is due to the fact that more than 95\% of the Bitcoin fees are clustered between $10^3$ and $10^5$, with no more than $10^5$ non-repeating values. The size of the dataset is 80,000, which already covers all non-repeating values. All possible fetches had been fed into the discriminator during the training process, so the discriminator can easily distinguish between real and fake data. Once the amount of embedded data is slightly larger (even 1 bit), the samples carrying data are easily recognized by the discriminator. The ability of the trained generator to embed data is thus poor.

Another phenomenon is that there is almost no difference between R-GAN and CCR-GAN in Ethereum amount in terms of data embedding ability. This is because when the magnitude of the largest data in the dataset reaches $10^{17}$, the computer precision error becomes the main error in the data decoding process. The essence of increasing the magnitude is to increase the precision of the rounding process, but it is limited by the precision of the computer floating-point numbers, which can only be accurately calculated to 17 decimal places at most.

\begin{tcolorbox}[colframe=black, colback=gray!10, boxrule=1pt,width=\linewidth]
  \textbf{Answer to RQ3:} It is only able to increase the magnitude of the dataset in a limited range (within $10^{17}$) to enhance the capacity. When the magnitude reaches $10^{17}$, the computational accuracy of the computer limits the accuracy of data decoding.
\end{tcolorbox}

\subsection{Concealment}

\textbf{We evaluate the concealment using CTR.} The concealment experiments follow the same settings as~\cite{wang2022practical}, except for the dataset size. We use larger datasets to ensure reliable results. We generate 80,000 transaction fileds using R-GAN and CCR-GAN, respectively. These data-carrying fields, along with an equal number of normal transaction fields, form the CTR dataset. 
For each configuration, we perform 10 experiments and averaged the results. Experiment results are summarized in Table~\ref{tab: CON}. 

For the same transaction field, the samples generated by CCR-GAN are more easily recognized than those generated by R-GAN. Technically, the fact that CCR-GAN does not eliminate very large and very small outliers results in an extremely uneven distribution of the normalized dataset used to train CCR-GAN. This in turn reduces the quality of the trained CCR-GAN model and hence the generated data is more easily recognized. In particular, in the R-GAN scheme for Ethereum amount, simply changing the distribution of the training dataset is also able to reduce the accuracy with which the generated samples are recognized (from 0.853 to 0.785). Intuitively, as more bits are used to store covert data, fewer bits are available to match the features of normal amounts, resulting in a higher accuracy of being recognized. 

For different transaction fields, the accuracy of being recognized varies even if the amount of embedded data is the same. For example, R-GAN for Bitcoin amount and CCR-GAN for Ethereum fee are both capable of embedding 24-bit data, while CCR-GAN for Ethereum fee is recognized with lower accuracy. This difference is caused by the distribution of the normalized training dataset. The more evenly the normalized dataset is distributed between 0 and 1, the higher the quality of the dataset, the better the R-GAN (or CCR-GAN) is trained, and the less accurate it is detected. In the CCR-GAN for Bitcoin amount, most of the amounts are between $10^{4}$ and $10^{7}$, and the normalized data is between $10^{-9}$ and $10^{-6}$. In the CCR-GAN for Ethereum fee, most of the fees are between $10^{9}$ and $10^{11}$, and the normalized data is between $10^{-4}$ and $10^{-2}$. The normalized Ethereum fee are more evenly distributed than the Bitcoin amount, so the less accuracy the generated samples is to be recognized. In particular, experiments with the R-GAN for Ethereum amount illustrate that when the distributions of the two normalized datasets are close to each other, the smaller the amount of data embedded, the lower the accuracy of being recognized. This is also consistent with the intuition that the more amount of embedded data the more likely it is to expose potential features. Furthermore, the Bitcoin fee is an exception because its number of possible values ($10,000$) is much smaller than the size of the training dataset of R-GAN (or CCR-GAN) ($80,000$), and the discriminator can label almost every sample in an enumerative manner.

\begin{tcolorbox}[colframe=black, colback=gray!10, boxrule=1pt,width=\linewidth]
  \textbf{Answer to RQ4:} Both R-GAN and CCR-GAN possess acceptable concealment. R-GAN provides higher concealment than CCR-GAN. Specifically, Bitcoin amount, Bitcoin fee, and Ethereum fee are recognized with accuracy of less than $0.8$, possessing a high degree of concealment. 
\end{tcolorbox}
\begin{tcolorbox}[colframe=black, colback=gray!10, boxrule=1pt,width=\linewidth]
  \textbf{Answer to RQ5:} T2C can effectively enhance concealment at the expense of capacity. Excluding computer accuracy limitations, for every 10-fold reduction in dataset magnitude, the capacity is reduced by 2-3 bits and the concealment is enhanced by 3\%-4\%.
\end{tcolorbox}

\begin{table}[htb]
\scriptsize
\renewcommand{\arraystretch}{1.3}
\captionsetup{justification=raggedright}
\caption{Recognition results of CTR. The dataset is composed of positive and negative samples with a ratio of 1:1, and is divided into a training set and a test set with a ratio of 7:3. Each set consists of half the normal field and the data-carrying field. The closer the result is to $0.5$, the more concealment the scheme possesses.}
\begin{center}
\begin{tabular}{c c c c c c}
\hline
\multirow{2}{*}{\textbf{Field}} & \multirow{2}{*}{\textbf{Scheme}} & \multirow{2}{*}{\textbf{\shortstack{Reduced \\ magnitude}}} &\multirow{2}{*}{\textbf{\shortstack{Dataset \\ magnitude}}} & \multirow{2}{*}{$\bm{m}$} & \multirow{2}{*}{\textbf{Accuracy}} \\
 & & & & & \\
\hline
\multirow{2}{*}{\shortstack{Bitcoin \\ amount}} & R-GAN   & / & $10^4\sim10^7$  & $11$ & $0.667$ \\ 
\cline{2-6}
                                                & CCR-GAN & / & $10^2\sim10^{13}$ & $24$ & $0.732$ \\
\hline
\multirow{2}{*}{\shortstack{Bitcoin \\ fee}} & R-GAN   & / & $10^2\sim10^4$  & $2$ & $0.795$ \\ 
\cline{2-6}
                                             & CCR-GAN & / & $10^2\sim10^7$  & $4$ & $0.796$ \\
\hline
\multirow{5}{*}{\shortstack{Ethereum \\ amount}} & \multirow{4}{*}{R-GAN}   & /      & $10^{14}\sim10^{18}$  & $41$ & $0.853$ \\ 
\cline{3-6}
                                                 &                          & $10^1$ & $10^{13}\sim10^{17}$  & $40$ & $0.851$ \\ 
\cline{3-6}                                                 
                                                 &                          & $10^2$ & $10^{12}\sim10^{16}$  & $38$ & $0.813$ \\ 
\cline{3-6}
                                                 &                          & $10^3$ & $10^{11}\sim10^{15}$  & $35$ & $0.785$ \\ 
\cline{2-6}
                                                 & CCR-GAN                  & /      & $10^{3}\sim10^{23}$   & $40$ & $0.968$ \\
\hline
\multirow{2}{*}{\shortstack{Ethereum \\ fee}} & R-GAN  & / & $10^9\sim10^{11}$  & $21$ & $0.577$ \\ 
\cline{2-6}
                                             & CCR-GAN & / & $10^9\sim10^{13}$  & $24$ & $0.607$ \\
\hline
\end{tabular}
\label{tab: CON}
\end{center}
\end{table}

\subsection{Effectiveness of ClipSigmoid} \label{sec: effect_of_clipsigmoid}

\textbf{We verify the effect of ClipSigmoid by plotting the loss function during training.} We take the Bitcoin amount as an example to verify the effectiveness of ClipSigmod. We keep all other parameters consistent except for the activation function and compare the change in model loss when using Sigmoid and ClipSigmoid as the activation function during the training. Fig. \ref{fig: effective} illustrates the impact of ClipSigmoid. At the 8th epoch, the model without ClipSigmoid (including both the generator and the discriminator) experiences a sudden spike and is difficult to converge again during subsequent training, while the model loss with ClipSigmoid begins to decrease until the model converges. 
\begin{figure}[t]
\centerline{\includegraphics[width=0.45\textwidth]{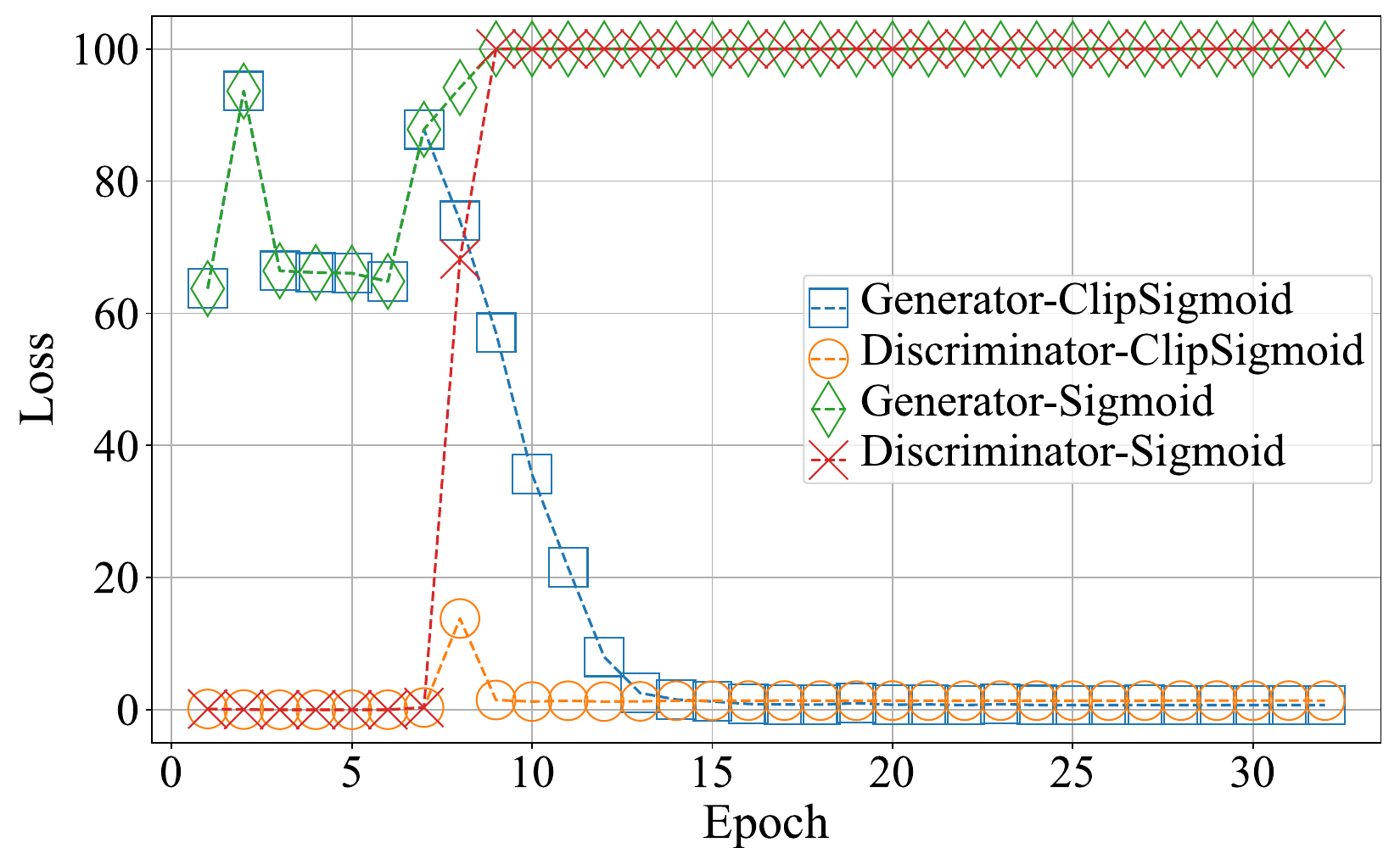}}
\captionsetup{justification=raggedright}
\caption{The effect of ClipSigmoid.}
\label{fig: effective}
\end{figure}

\begin{tcolorbox}[colframe=black, colback=gray!10, boxrule=1pt,width=\linewidth]
  \textbf{Answer to RQ6:} By setting up the above comparison experiments, it can be noticed that the training loss of the model is reduced after using CilpSigmoid, while the model with Sigmoid does not converge.
\end{tcolorbox}

\subsection{Comparison}

We use $\mathsf{CER}$ (see Equation~(\ref{equ: CER})) to evaluate the capacity expansion capability of the proposed schemes against blockchain-based covert channels that do not utilize the transaction amount or the transaction fee as the embedding field. For covert channels that utilize the Bitcoin transaction amount as the embedding field, we compare $\mathsf{AC}$ and concealment to demonstrate the superiority of our schemes. We only compare schemes using Bitcoin amount as the embedding field, as there is little research employing Ethereum amount and transaction fee as the embedding field.

\subsubsection{Comparison of the capacity expansion capability}

In Bitcoin, we consider that the transaction is a P2PKH\footnote{Pay-to-Public-Key-Hash, a type of ScriptPubKey which locks Bitcoin to the hash of a public key~\cite{p2pkh}.} transaction with one input and two outputs by default, as this type of transaction accounts for the largest number~\cite{transactionsfeeinfo}. In Ethereum, we consider that the transaction is a transfer transaction with one input and one output. The $\mathsf{AC}$ of both the proposed schemes and baselines is calculated using the above setting. We discard the CCR-GAN scheme for Ethereum amount and choose a reduced magnitude of $10^3$ as the criterion in the R-GAN scheme for Ethereum amount. The reason is that schemes with recognized accuracy greater than $0.8$ are too low in concealment. Table~\ref{tab: CER} presents $\mathsf{CER}$ of proposed schemes. It can be observed that R-GAN can boost capacity up to 291.67\% of baselines and CCR-GAN up to 200\% of baselines. Although the transaction fields generated by R-GAN and CCR-GAN may slightly increase the identification risk, the huge capacity increase makes the risk worthwhile.

\begin{table}[htb]
\scriptsize
\renewcommand{\arraystretch}{1.3}
\captionsetup{justification=raggedright}
\caption{$\mathsf{CER}$ of proposed schemes.}
\begin{center}
\begin{tabular}{c c c c c c}
\hline
\textbf{Field} & \textbf{Scheme} & \textbf{HC-CDE} & \textbf{DSA} & \textbf{Un-UTXO} & \textbf{DLchain} \\
\hline
\multirow{2}{*}{\shortstack{Bitcoin \\ amount}}  & R-GAN   & $91.67\%$ & $4.30\%$ & $6.88\%$ & $8.59\%$ \\
\cline{2-6}
                                                 & CCR-GAN & $200.00\%$ & $9.38\%$ & $15.00\%$ & $18.75\%$ \\
\hline
\multirow{2}{*}{\shortstack{Bitcoin \\ fee}}     & R-GAN   & $16.67\%$ & $0.78\%$ & $1.25\%$ & $1.56\%$ \\
\cline{2-6}
                                                 & CCR-GAN & $33.33\%$ & $1.56\%$ & $2.50\%$ & $3.13\%$ \\
\hline
\multirow{2}{*}{\shortstack{Ethereum \\ amount}} & R-GAN   & $291.67\%$ & $13.67\%$ & $21.88\%$ & $27.34\%$ \\
\cline{2-6}
                                                 & CCR-GAN & $/      $ & $     / $ & $/     $ & $/     $ \\
\hline
\multirow{2}{*}{\shortstack{Ethereum \\ fee}}    & R-GAN   & $175.00\%$ & $8.20\%$ & $13.13\%$ & $16.41\%$ \\
\cline{2-6}
                                                 & CCR-GAN & $200.00\%$ & $9.38\%$ & $15.00\%$ & $18.75\%$ \\
\hline
\end{tabular}
\label{tab: CER}
\end{center}
\end{table}

\begin{table}[htb]
\scriptsize
\captionsetup{justification=raggedright}
\caption{Comparison of Bitcoin amount embedding.}
\begin{center}
\begin{tabular}{c c c c c c}
\hline
\textbf{Scheme} & \textbf{$\mathsf{AC}$} & \textbf{Accuracy} & \textbf{Precision} & \textbf{Recall} & \textbf{F1-score} \\
\hline
CCMBBT  & $23$ & $0.909$ & $0.912$ & $0.909$ & $0.909$ \\
STCBC   & $\sim27$ & $0.911$ & $0.913$ & $0.911$ & $0.911$ \\
AMASC   & $ 8$ & $0.975$ & $0.976$ & $0.975$ & $0.975$ \\
R-GAN   & $11$ & $0.667$ & $0.676$ & $0.667$ & $0.663$ \\
CCR-GAN & $24$ & $0.732$ & $0.733$ & $0.732$ & $0.732$ \\
\hline
\end{tabular}
\label{tab: comparison}
\end{center}
\end{table}

\subsubsection{Comparison of Bitcoin amount embedding}

We also assume that Bitcoin transactions are P2PKH transactions with one input and two outputs. We perform the concealment experiments with the dataset size of 16,000. Comparison results are shown in Table~\ref{tab: comparison}. Compared to CCMBBT, R-GAN has a smaller capacity and exhibits higher concealment. CCR-GAN has a larger capacity to CCMBBT, while its concealment surpasses that of CCMBBT. This is because CCMBBT simply encodes covert data as a number using the ASCII encoding, making their amounts easily distinguishable. R-GAN trades capacity for higher concealment and outperforms CCMBBT in concealment, while CCR-GAN outperforms CCMBBT in both capacity and concealment. The recognized accuracy of STCBC is more than 0.9. Although R-GAN and CCR-GAN have slightly smaller capacity than STCBC, they possess far stronger concealment than STCBC. Both R-GAN and CCR-GAN outperform AMASC in capacity and concealment.

\begin{tcolorbox}[colframe=black, colback=gray!10, boxrule=1pt,width=\linewidth]
  \textbf{Answer to RQ7:} In terms of capacity expansion, R-GAN can boost capacity up to 291.67\% of baselines and CCR-GAN up to 200\% of baselines. In terms of capacity and concealment, R-GAN and CCR-GAN are able to achieve higher capacity with guaranteed concealment, whereas the existing schemes are difficult to satisfy both high capacity and concealment.
\end{tcolorbox}

\section{Related Work}

\noindent \textbf{Blockchain-based covert channels.} In addition to baselines, several works are relevant to blockchain-based covert channels. Partala~\cite{partala2018provably} first builds a covert channel in the blockchain. They propose BLOCCE and demonstrate its security. BLOCCE stores 1-bit covert data into the least significant bit of the address. Gao~\textit{et al.}~\cite{gao2020achieving} and Zhang~\textit{et al.}~\cite{zhang2023ebdl} utilize OP\_RETURN to encode covert data. Zhang~\textit{et al.}~\cite{zhang2021research} construct a covert channel based on the parameters of Ethereum smart contracts. Luo~\textit{et al.}~\cite{luo2021novel} represent bits 0 and 1 by the presence or absence of transactions between addresses, and also encode covert data into the transaction amount. Zhang~\textit{et al.}~\cite{zhang2021approach} employ the Ethereum Whisper protocol to establish a covert channe. Alsalami~\textit{et al.}~\cite{alsalami2020uncontrolled} explore randomness in the blockchain to build covert channels. None of them consider generating required transaction fields, whereas our approaches specifically focus on generating these other fields.

\noindent \textbf{Transaction generation via AI.} Wang \textit{et al.}~\cite{wang2022practical} propose a PCTC model for generating transaction fields. They aim to generate indistinguishable transaction fields and provide a reference for creating covert transactions. Liu~\textit{et al.}~\cite{liu2020research} employ GAN to generate the Ethereum transaction fields and embed covert data. However, the receiver may not be unable to obtain decoded data consistent with the original covert data. Researchers also explore automating smart contract generation using AI~\cite{shao2020lsc, dwivedi2022auto, mavridou2018designing, hu2022scsguard, He2023msdc}, while none of them consider embedding covert data during the generation process. There is currently no research on embedding data while generating transaction fields via AI, whereas our approaches address this problem.

\noindent \textbf{AI-based text steganography.} 
The proposed schemes can also be considered as a form of text steganography. 
Existing research on text steganography focuses on linguistic steganography, where covert data is hidden within language semantics. Yang~\textit{et al.}~\cite{yang2021vae} utilize variational auto-encoders to hide covert data into word selection. Their subsequent work enhances the concealment of generated sentences~\cite{yang2021linguistic}. Li~\textit{et al.}~\cite{li2021topic} achieve steganographic long text generation by encoding covert data into entities and relationships within a knowledge graph. Zhou~\textit{et al.}~\cite{zhou2021linguistic} propose an adaptive embedding algorithm with a similarity function to implement linguistic steganography, ensuring that the embedded distribution remains consistent with the actual distribution. These methods do not apply to embedding covert data into purely numerical transaction fields.

\section{Conclusion}

In this paper, we have introduced GBSF, a generic framework for expanding the channel capacity of blockchain-based steganography. GBSF involves the sender generating indistinguishable required fields while embedding covert data. We have presented the R-GAN scheme, which utilizes GAN with a reversible generator to generate the required fields and encodes covert data as input noise to the GAN generator. Additionally, we have proposed CCR-GAN as an enhancement to R-GAN, and develop T2C to balance capacity and concealment. Experimental results have demonstrated that our proposed schemes outperform baselines.




\bibliographystyle{IEEEtran}
\bibliography{IEEEabrv,mybib.bib}

\begin{thebibliography}{10}
\providecommand{\url}[1]{#1}
\csname url@samestyle\endcsname
\providecommand{\newblock}{\relax}
\providecommand{\bibinfo}[2]{#2}
\providecommand{\BIBentrySTDinterwordspacing}{\spaceskip=0pt\relax}
\providecommand{\BIBentryALTinterwordstretchfactor}{4}
\providecommand{\BIBentryALTinterwordspacing}{\spaceskip=\fontdimen2\font plus
\BIBentryALTinterwordstretchfactor\fontdimen3\font minus
  \fontdimen4\font\relax}
\providecommand{\BIBforeignlanguage}[2]{{%
\expandafter\ifx\csname l@#1\endcsname\relax
\typeout{** WARNING: IEEEtran.bst: No hyphenation pattern has been}%
\typeout{** loaded for the language `#1'. Using the pattern for}%
\typeout{** the default language instead.}%
\else
\language=\csname l@#1\endcsname
\fi
#2}}
\providecommand{\BIBdecl}{\relax}
\BIBdecl

\bibitem{mandal2022digital}
P.~C. Mandal, I.~Mukherjee, G.~Paul, and B.~Chatterji, ``Digital image
  steganography: A literature survey,'' \emph{Information Sciences}, 2022.

\bibitem{liu2022deep}
Q.~Liu, J.~Yang, H.~Jiang, J.~Wu, T.~Peng, T.~Wang, and G.~Wang, ``When deep
  learning meets steganography: Protecting inference privacy in the dark,'' in
  \emph{IEEE INFOCOM 2022-IEEE Conference on Computer Communications}.\hskip
  1em plus 0.5em minus 0.4em\relax IEEE, 2022, pp. 590--599.

\bibitem{zheng2022wmdefense}
X.~Zheng, Q.~Dong, and A.~Fu, ``Wmdefense: Using watermark to defense byzantine
  attacks in federated learning,'' in \emph{IEEE INFOCOM 2022-IEEE Conference
  on Computer Communications Workshops (INFOCOM WKSHPS)}.\hskip 1em plus 0.5em
  minus 0.4em\relax IEEE, 2022, pp. 1--6.

\bibitem{qiang2023natural}
J.~Qiang, S.~Zhu, Y.~Li, Y.~Zhu, Y.~Yuan, and X.~Wu, ``Natural language
  watermarking via paraphraser-based lexical substitution,'' \emph{Artificial
  Intelligence}, vol. 317, p. 103859, 2023.

\bibitem{rosen2021balboa}
M.~B. Rosen, J.~Parker, and A.~J. Malozemoff, ``Balboa: Bobbing and weaving
  around network censorship,'' in \emph{30th USENIX Security Symposium (USENIX
  Security 21)}, 2021, pp. 3399--3413.

\bibitem{zhou2023generative}
Z.~Zhou, X.~Dong, R.~Meng, M.~Wang, H.~Yan, K.~Yu, and K.-K.~R. Choo,
  ``Generative steganography via auto-generation of semantic object contours,''
  \emph{IEEE Transactions on Information Forensics and Security}, 2023.

\bibitem{xu2022robust}
Y.~Xu, C.~Mou, Y.~Hu, J.~Xie, and J.~Zhang, ``Robust invertible image
  steganography,'' in \emph{Proceedings of the IEEE/CVF Conference on Computer
  Vision and Pattern Recognition}, 2022, pp. 7875--7884.

\bibitem{wang2024generative2}
J.~Wang, H.~Du, D.~Niyato, J.~Kang, S.~Cui, X.~S. Shen, and P.~Zhang,
  ``Generative ai for integrated sensing and communication: Insights from the
  physical layer perspective,'' \emph{IEEE Wireless Communications}, 2024.

\bibitem{zhang2023covert}
T.~Zhang, B.~Li, Y.~Zhu, T.~Han, and Q.~Wu, ``Covert channels in blockchain and
  blockchain based covert communication: Overview, state-of-the-art, and future
  directions,'' \emph{Computer Communications}, 2023.

\bibitem{miao2022privacy}
Y.~Miao, Z.~Liu, H.~Li, K.-K.~R. Choo, and R.~H. Deng, ``Privacy-preserving
  byzantine-robust federated learning via blockchain systems,'' \emph{IEEE
  Transactions on Information Forensics and Security}, vol.~17, pp. 2848--2861,
  2022.

\bibitem{chen2022blockchain}
Z.~Chen, L.~Zhu, P.~Jiang, C.~Zhang, F.~Gao, J.~He, D.~Xu, and Y.~Zhang,
  ``Blockchain meets covert communication: A survey,'' \emph{IEEE
  Communications Surveys \& Tutorials}, 2022.

\bibitem{du2022applications}
B.~Du, D.~He, M.~Luo, C.~Peng, and Q.~Feng, ``The applications of blockchain in
  the covert communication,'' \emph{Wireless Communications and Mobile
  Computing}, vol. 2022, 2022.

\bibitem{wang2023covert}
Z.~Wang, L.~Zhang, R.~Guo, G.~Wang, J.~Qiu, S.~Su, Y.~Liu, G.~Xu, and Z.~Tian,
  ``A covert channel over blockchain based on label tree without long waiting
  times,'' \emph{Computer Networks}, p. 109843, 2023.

\bibitem{zhang2023ebdl}
C.~Zhang, L.~Zhu, C.~Xu, Z.~Zhang, and R.~Lu, ``Ebdl: Effective
  blockchain-based covert storage channel with dynamic labels,'' \emph{Journal
  of Network and Computer Applications}, vol. 210, p. 103541, 2023.

\bibitem{wang2022practical}
M.~Wang, Z.~Zhang, J.~He, F.~Gao, M.~Li, S.~Xu, and L.~Zhu, ``Practical
  blockchain-based steganographic communication via adversarial ai: A case
  study in bitcoin,'' \emph{The Computer Journal}, vol.~65, no.~11, pp.
  2926--2938, 2022.

\bibitem{creswell2018generative}
A.~Creswell, T.~White, V.~Dumoulin, K.~Arulkumaran, B.~Sengupta, and A.~A.
  Bharath, ``Generative adversarial networks: An overview,'' \emph{IEEE signal
  processing magazine}, vol.~35, no.~1, pp. 53--65, 2018.

\bibitem{wang2024generative}
J.~Wang, H.~Du, Y.~Liu, G.~Sun, D.~Niyato, S.~Mao, D.~I. Kim, and X.~Shen,
  ``Generative ai based secure wireless sensing for isac networks,''
  \emph{arXiv preprint arXiv:2408.11398}, 2024.

\bibitem{koptyra2022extension}
K.~Koptyra and M.~R. Ogiela, ``An extension of imagechain concept that allows
  multiple images per block,'' in \emph{IEEE INFOCOM 2022-IEEE Conference on
  Computer Communications Workshops (INFOCOM WKSHPS)}.\hskip 1em plus 0.5em
  minus 0.4em\relax IEEE, 2022, pp. 1--2.

\bibitem{zhang2020magview}
J.~Zhang, X.~Ji, W.~Xu, Y.-C. Chen, Y.~Tang, and G.~Qu, ``Magview: A
  distributed magnetic covert channel via video encoding and decoding,'' in
  \emph{IEEE INFOCOM 2020-IEEE Conference on Computer Communications}.\hskip
  1em plus 0.5em minus 0.4em\relax IEEE, 2020, pp. 357--366.

\bibitem{yang2023semantic}
T.~Yang, H.~Wu, B.~Yi, G.~Feng, and X.~Zhang, ``Semantic-preserving linguistic
  steganography by pivot translation and semantic-aware bins coding,''
  \emph{IEEE Transactions on Dependable and Secure Computing}, 2023.

\bibitem{yang2022linguistic}
J.~Yang, Z.~Yang, J.~Zou, H.~Tu, and Y.~Huang, ``Linguistic steganalysis toward
  social network,'' \emph{IEEE Transactions on Information Forensics and
  Security}, vol.~18, pp. 859--871, 2022.

\bibitem{li2022gasla}
J.~Li, Y.~Liu, W.~Xu, and Z.~Li, ``Gasla: Enhancing the applicability of sign
  language translation,'' in \emph{IEEE INFOCOM 2022-IEEE Conference on
  Computer Communications}.\hskip 1em plus 0.5em minus 0.4em\relax IEEE, 2022,
  pp. 1249--1258.

\bibitem{cui2019unseencode}
H.~Cui, H.~Bian, W.~Zhang, and N.~Yu, ``Unseencode: Invisible on-screen barcode
  with image-based extraction,'' in \emph{IEEE INFOCOM 2019-IEEE Conference on
  Computer Communications}.\hskip 1em plus 0.5em minus 0.4em\relax IEEE, 2019,
  pp. 1315--1323.

\bibitem{he2024preventing}
J.~He, J.~Wang, N.~Wang, S.~Guo, L.~Zhu, D.~Niyato, and T.~Xiang, ``Preventing
  non-intrusive load monitoring privacy invasion: A precise adversarial attack
  scheme for networked smart meters,'' \emph{arXiv preprint arXiv:2412.16893},
  2024.

\bibitem{he2025advancing}
J.~He, J.~Feng, S.~Guo, Z.~Chen, Y.~Liu, T.~Xiang, and L.~Zhu, ``Advancing
  non-intrusive load monitoring: Predicting appliance-level power consumption
  with indirect supervision,'' \emph{IEEE Transactions on Network Science and
  Engineering}, 2025.

\bibitem{guo2024single}
C.~Guo, Z.~Jia, Y.~Lin, and H.~Yang, ``Single-precision floating-point adder
  based on ieee 754 standard through verilog,'' in \emph{International
  Conference on Artificial Intelligence and Communication Technology}.\hskip
  1em plus 0.5em minus 0.4em\relax Springer, 2024, pp. 141--150.

\bibitem{yuan2025blockchain}
X.~Yuan, P.~Jiang, Z.~Chen, C.~Zhang, F.~Gao, and L.~Zhu, ``Blockchain-based
  group covert communication for iot network,'' \emph{IEEE Internet of Things
  Journal}, 2025.

\bibitem{chen2025tackling}
Z.~Chen, L.~Zhu, P.~Jiang, J.~He, and Z.~Zhang, ``Tackling data mining risks: A
  tripartite covert channel merging blockchain and ipfs,'' \emph{IEEE
  Transactions on Network Science and Engineering}, 2025.

\bibitem{zhou2022secret}
Z.~Zhou, Y.~Su, J.~Li, K.~Yu, Q.~J. Wu, Z.~Fu, and Y.~Shi, ``Secret-to-image
  reversible transformation for generative steganography,'' \emph{IEEE
  Transactions on Dependable and Secure Computing}, vol.~20, no.~5, pp.
  4118--4134, 2022.

\bibitem{qi2024provably}
Y.~Qi, K.~Chen, K.~Zeng, W.~Zhang, and N.~Yu, ``Provably secure disambiguating
  neural linguistic steganography,'' \emph{IEEE Transactions on Dependable and
  Secure Computing}, 2024.

\bibitem{katz2007introduction}
J.~Katz and Y.~Lindell, \emph{Introduction to modern cryptography: principles
  and protocols}.\hskip 1em plus 0.5em minus 0.4em\relax Chapman and hall/CRC,
  2007.

\bibitem{butora2023errorless}
J.~Butora, P.~Puteaux, and P.~Bas, ``Errorless robust jpeg steganography using
  outputs of jpeg coders,'' \emph{IEEE Transactions on Dependable and Secure
  Computing}, vol.~21, no.~4, pp. 2394--2406, 2023.

\bibitem{pytorch}
\BIBentryALTinterwordspacing
Pytorch versions installation. Accessed: June 8, 2025. [Online]. Available:
  \url{https://pytorch.org/get-started/previous-versions/}
\BIBentrySTDinterwordspacing

\bibitem{cao2020chain}
H.~Cao, H.~Yin, F.~Gao, Z.~Zhang, B.~Khoussainov, S.~Xu, and L.~Zhu,
  ``Chain-based covert data embedding schemes in blockchain,'' \emph{IEEE
  Internet of Things Journal}, vol.~9, no.~16, pp. 14\,699--14\,707, 2020.

\bibitem{fionov2019exploring}
A.~Fionov, ``Exploring covert channels in bitcoin transactions,'' in \emph{2019
  International Multi-Conference on Engineering, Computer and Information
  Sciences (SIBIRCON)}.\hskip 1em plus 0.5em minus 0.4em\relax IEEE, 2019, pp.
  0059--0064.

\bibitem{gregoriadis2022analysis}
M.~Gregoriadis, R.~Muth, and M.~Florian, ``Analysis of arbitrary content on
  blockchain-based systems using bigquery,'' in \emph{Companion Proceedings of
  the Web Conference 2022}, 2022, pp. 478--487.

\bibitem{tian2020dlchain}
J.~Tian, G.~Gou, C.~Liu, Y.~Chen, G.~Xiong, and Z.~Li, ``Dlchain: A covert
  channel over blockchain based on dynamic labels,'' in \emph{Information and
  Communications Security: 21st International Conference, ICICS 2019, Beijing,
  China, December 15--17, 2019, Revised Selected Papers 21}.\hskip 1em plus
  0.5em minus 0.4em\relax Springer, 2020, pp. 814--830.

\bibitem{luo2021novel}
X.~Luo, P.~Zhang, M.~Zhang, H.~Li, and Q.~Cheng, ``A novel covert communication
  method based on bitcoin transaction,'' \emph{IEEE Transactions on Industrial
  Informatics}, vol.~18, no.~4, pp. 2830--2839, 2021.

\bibitem{zhang2024blockchain}
P.~Zhang, Q.~Cheng, M.~Zhang, and X.~Luo, ``A blockchain-based secure covert
  communication method via shamir threshold and stc mapping,'' \emph{IEEE
  Transactions on Dependable and Secure Computing}, 2024.

\bibitem{tian2024amount}
Y.~Tian, X.~Liao, L.~Dong, Y.~Xu, and H.~Jiang, ``Amount-based covert
  communication over blockchain,'' \emph{IEEE Transactions on Network and
  Service Management}, 2024.

\bibitem{p2pkh}
\BIBentryALTinterwordspacing
Pay-to-pubkey hash. Accessed: June 8, 2025. [Online]. Available:
  \url{http://bitcoinwiki.org/wiki/pay-to-pubkey-hash/}
\BIBentrySTDinterwordspacing

\bibitem{transactionsfeeinfo}
\BIBentryALTinterwordspacing
Bitcoin protocol layer statistics. Accessed: June 1, 2023. [Online]. Available:
  \url{https://transactionfee.info/}
\BIBentrySTDinterwordspacing

\bibitem{partala2018provably}
J.~Partala, ``Provably secure covert communication on blockchain,''
  \emph{Cryptography}, vol.~2, no.~3, p.~18, 2018.

\bibitem{gao2020achieving}
F.~Gao, L.~Zhu, K.~Gai, C.~Zhang, and S.~Liu, ``Achieving a covert channel over
  an open blockchain network,'' \emph{IEEE Network}, vol.~34, no.~2, pp. 6--13,
  2020.

\bibitem{zhang2021research}
L.~Zhang, Z.~Zhang, W.~Wang, Z.~Jin, Y.~Su, and H.~Chen, ``Research on a covert
  communication model realized by using smart contracts in blockchain
  environment,'' \emph{IEEE Systems Journal}, vol.~16, no.~2, pp. 2822--2833,
  2021.

\bibitem{zhang2021approach}
L.~Zhang, Z.~Zhang, Z.~Jin, Y.~Su, and Z.~Wang, ``An approach of covert
  communication based on the ethereum whisper protocol in blockchain,''
  \emph{International Journal of Intelligent Systems}, vol.~36, no.~2, pp.
  962--996, 2021.

\bibitem{alsalami2020uncontrolled}
N.~Alsalami and B.~Zhang, ``Uncontrolled randomness in blockchains: Covert
  bulletin board for illicit activity,'' in \emph{2020 IEEE/ACM 28th
  International Symposium on Quality of Service (IWQoS)}.\hskip 1em plus 0.5em
  minus 0.4em\relax IEEE, 2020, pp. 1--10.

\bibitem{liu2020research}
J.~Liu, ``Research on information hiding methods based on blockchain
  technology,'' Master's thesis, Nanjing University of Information Science and
  Technology, 2020.

\bibitem{shao2020lsc}
W.~Shao, Z.~Wang, X.~Wang, K.~Qiu, C.~Jia, and C.~Jiang, ``Lsc: Online
  auto-update smart contracts for fortifying blockchain-based log systems,''
  \emph{Information Sciences}, vol. 512, pp. 506--517, 2020.

\bibitem{dwivedi2022auto}
V.~Dwivedi and A.~Norta, ``Auto-generation of smart contracts from a
  domain-specific xml-based language,'' in \emph{Intelligent Data Engineering
  and Analytics: Proceedings of the 9th International Conference on Frontiers
  in Intelligent Computing: Theory and Applications (FICTA 2021)}.\hskip 1em
  plus 0.5em minus 0.4em\relax Springer, 2022, pp. 549--564.

\bibitem{mavridou2018designing}
A.~Mavridou and A.~Laszka, ``Designing secure ethereum smart contracts: A
  finite state machine based approach,'' in \emph{Financial Cryptography and
  Data Security: 22nd International Conference, FC 2018, Nieuwpoort,
  Cura{\c{c}}ao, February 26--March 2, 2018, Revised Selected Papers 22}.\hskip
  1em plus 0.5em minus 0.4em\relax Springer, 2018, pp. 523--540.

\bibitem{hu2022scsguard}
H.~Hu, Q.~Bai, and Y.~Xu, ``Scsguard: Deep scam detection for ethereum smart
  contracts,'' in \emph{IEEE INFOCOM 2022-IEEE Conference on Computer
  Communications Workshops (INFOCOM WKSHPS)}.\hskip 1em plus 0.5em minus
  0.4em\relax IEEE, 2022, pp. 1--6.

\bibitem{He2023msdc}
J.~He, J.~Liu, Z.~Zhang, Y.~Chen, Y.~Liu, B.~Khoussainov, and L.~Zhu, ``Msdc:
  Exploiting multi-state power consumption in non-intrusive load monitoring
  based on a dual-cnn model,'' \emph{Proceedings of the AAAI Conference on
  Artificial Intelligence}, vol.~37, no.~4, pp. 5078--5086, 2023.

\bibitem{yang2021vae}
Z.~Yang, S.~Zhang, Y.~Hu, Z.~Hu, and Y.~Huang, ``Vae-stega: Linguistic
  steganography based on variational auto-encoder,'' \emph{IEEE Transactions on
  Information Forensics and Security}, vol.~16, 2021.

\bibitem{yang2021linguistic}
Z.~Yang, L.~Xiang, S.~Zhang, X.~Sun, and Y.~Huang, ``Linguistic generative
  steganography with enhanced cognitive-imperceptibility,'' \emph{IEEE Signal
  Processing Letters}, vol.~28, pp. 409--413, 2021.

\bibitem{li2021topic}
Y.~Li, J.~Zhang, Z.~Yang, and R.~Zhang, ``Topic-aware neural linguistic
  steganography based on knowledge graphs,'' \emph{ACM/IMS Transactions on Data
  Science}, vol.~2, no.~2, pp. 1--13, 2021.

\bibitem{zhou2021linguistic}
X.~Zhou, W.~Peng, B.~Yang, J.~Wen, Y.~Xue, and P.~Zhong, ``Linguistic
  steganography based on adaptive probability distribution,'' \emph{IEEE
  Transactions on Dependable and Secure Computing}, vol.~19, no.~5, pp.
  2982--2997, 2021.

\end{thebibliography}

\begin{IEEEbiography}[{\includegraphics[width=1in,height=1.25in,clip,keepaspectratio]{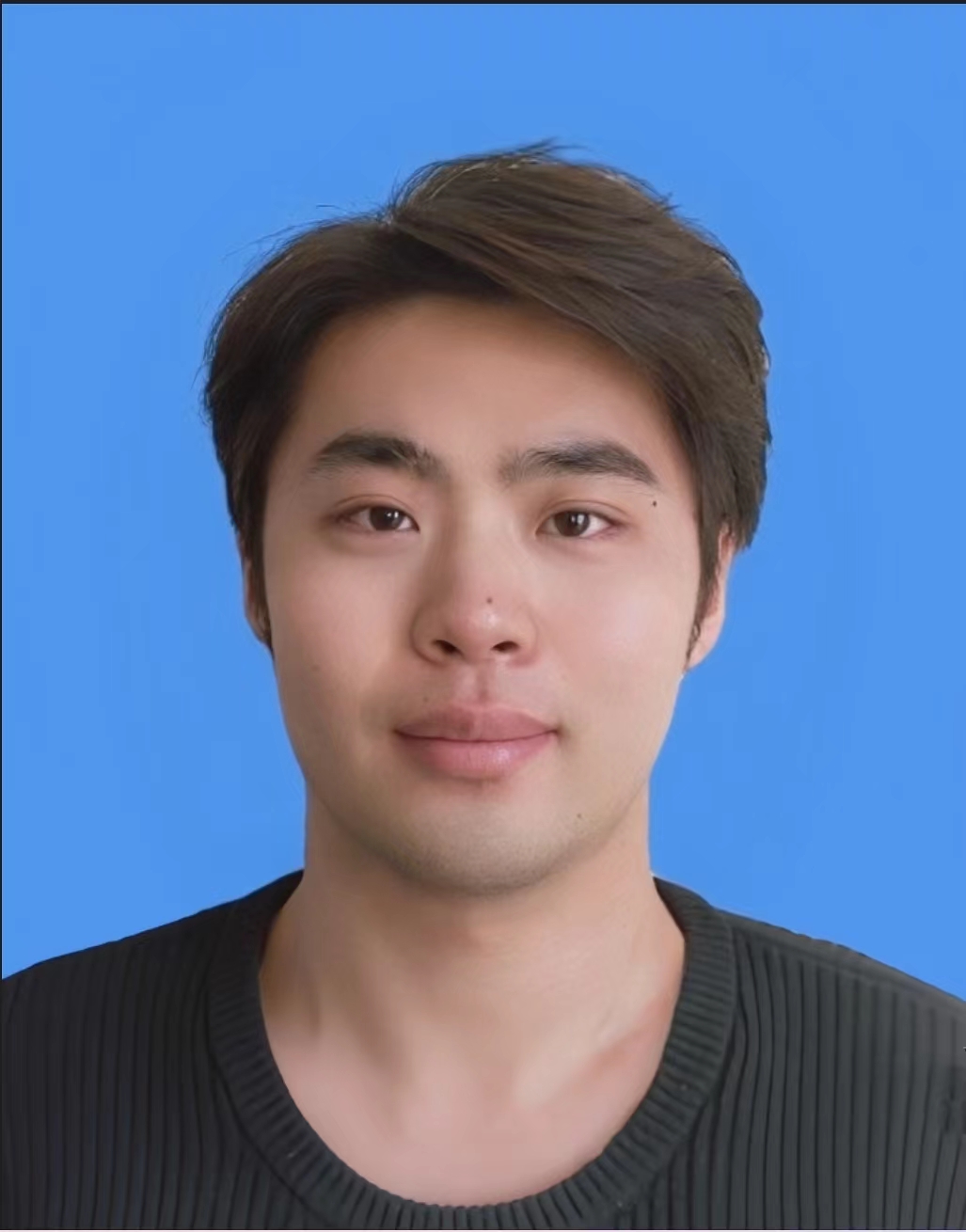}}]{Zhuo Chen}
received the B.E. degree in information security from the North China Electric Power University, Beijing, China, in 2019. He is currently pursuing the Ph.D. degree with the School of Cyberspace Science and Technology, Beijing Institute of Technology. His current research interests include blockchain technology and covert communication.
\end{IEEEbiography}

\begin{IEEEbiography}[{\includegraphics[width=1in,height=1.25in,clip,keepaspectratio]{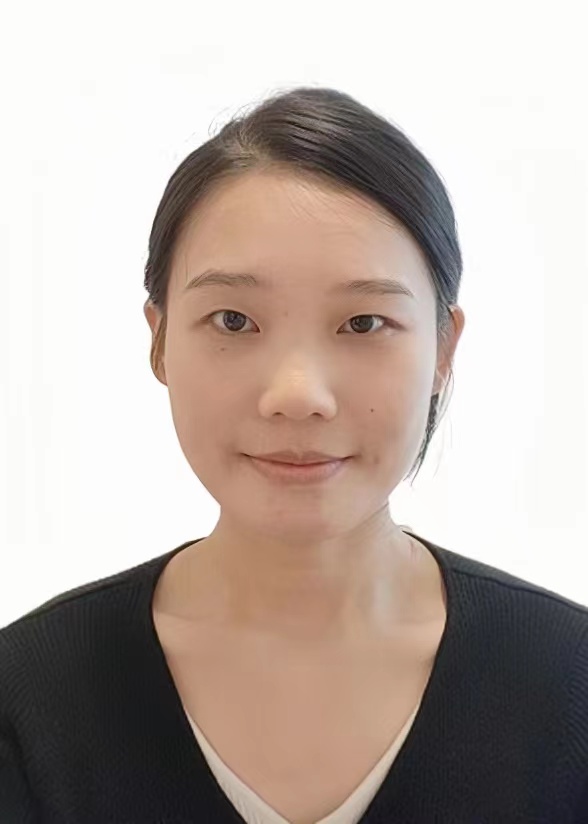}}]{Jialing He}
(Member, IEEE) received the M.S. and Ph.D. degrees from the Beijing Institute of Technology, Beijing, China, in 2018 and 2022, respectively, where she is currently an assistant research scientist in college of computer science, Chongqing University, Chongqing, China. Her current research interests include differential privacy and user behavior mining. 
\end{IEEEbiography}

\begin{IEEEbiography}[{\includegraphics[width=1in,height=1.25in,clip,keepaspectratio]{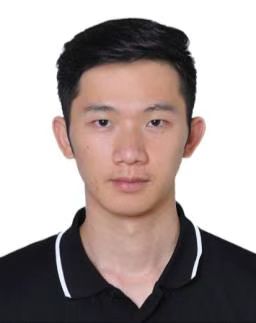}}]{Jiacheng Wang}
received the Ph.D. degree from the School of Communication and Information Engineering, Chongqing University of Posts and Telecommunications, Chongqing, China. He is currently a Research Associate in computer science and engineering with Nanyang Technological University, Singapore. His research interests include wireless sensing, semantic communications, and metaverse.
\end{IEEEbiography}

\begin{IEEEbiography}[{\includegraphics[width=1in,height=1.25in,clip,keepaspectratio]{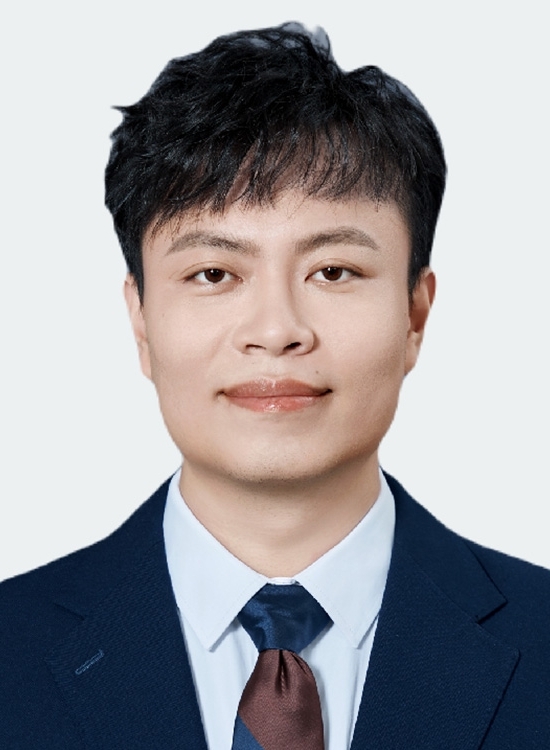}}]{Zehui Xiong}
received the B.Eng. degree in telecommunications engineering from the Huazhong University of Science and Technology (HUST), Wuhan, China, and the Ph.D. degree in computer science and engineering from Nanyang Technological University (NTU), Singapore. He is currently an Assistant Professor with the Singapore University of Technology and Design (SUTD), and also an Honorary Adjunct Senior Research Scientist with the Alibaba-NTU Singapore Joint Research Institute, Singapore. His research interests include wireless communications, the Internet of Things, blockchain, edge intelligence, and metaverse.
\end{IEEEbiography}

\begin{IEEEbiography}[{\includegraphics[width=1in,height=1.25in,clip,keepaspectratio]{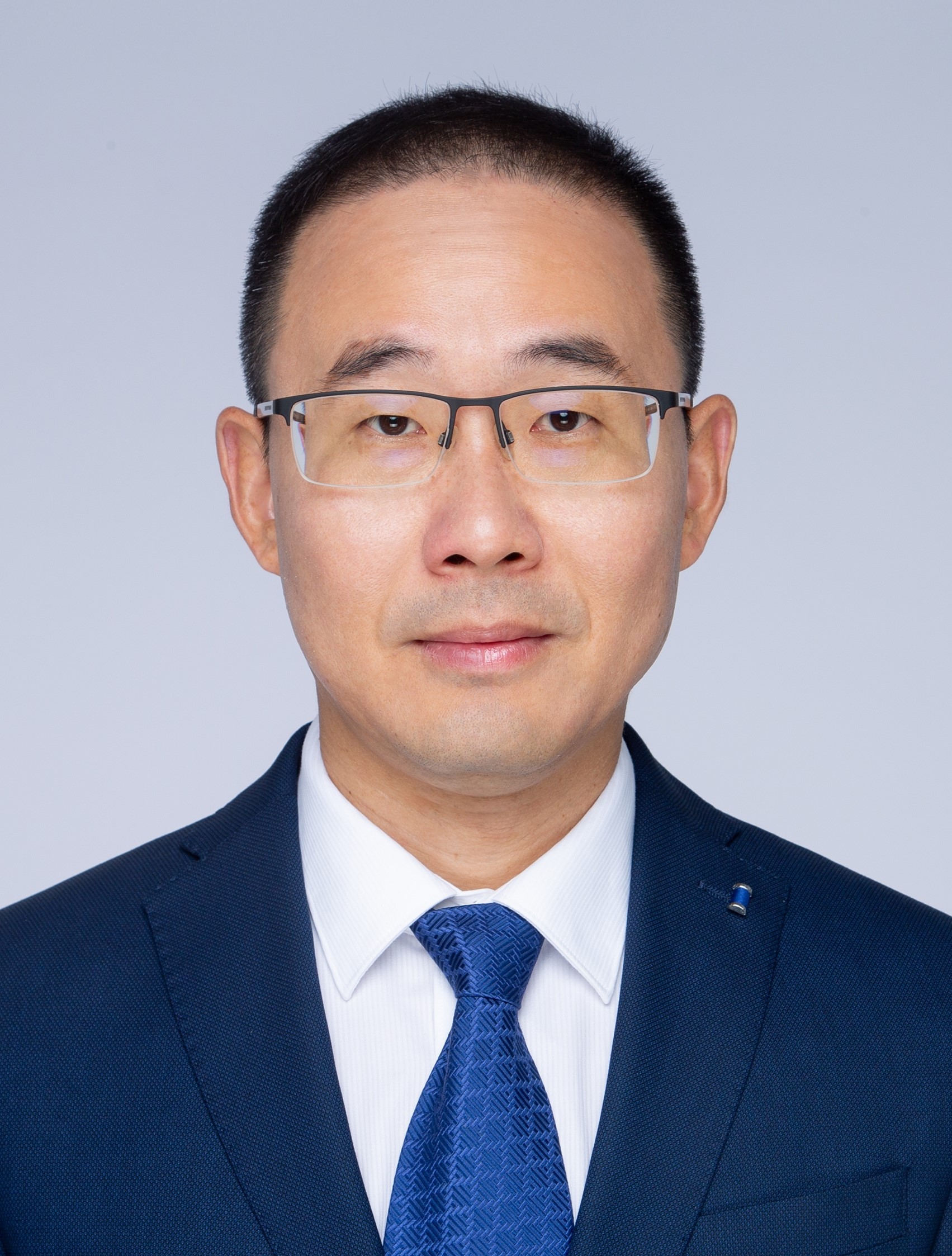}}]{Tao Xiang}
received the BEng, MS and PhD degrees in computer science from Chongqing University, China, in 2003, 2005, and 2008, respectively. He is currently a Professor of the College of Computer Science at Chongqing University. Prof. Xiang’s research interests include multimedia security, cloud security, data privacy and cryptography. He has published over 100 papers on international journals and conferences. He also served as a referee for numerous international journals and conferences.
\end{IEEEbiography}

\begin{IEEEbiography}[{\includegraphics[width=1in,height=1.25in,clip,keepaspectratio]{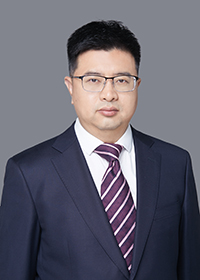}}]{Liehuang Zhu}
(Senior Member, IEEE) is a Full Professor with the School of Cyberspace Science and Technology, Beijing Institute of Technology. He is selected into the Program for New Century Excellent Talents in University from Ministry of Education, China. He has published over 60 SCI-indexed research papers in these areas, as well as a book published by Springer. His research interests include Internet of Things, cloud computing security, Internet, and mobile security. He serves on the editorial boards of three international journals, including IEEE Internet of Things Journal, IEEE Network, and IEEE Transactions on Vehicular Technology. He won the Best Paper Award at IEEE/ACM IWQoS 2017 and IEEE TrustCom 2018.
\end{IEEEbiography}

\begin{IEEEbiography}[{\includegraphics[width=1in,height=1.25in,clip,keepaspectratio]{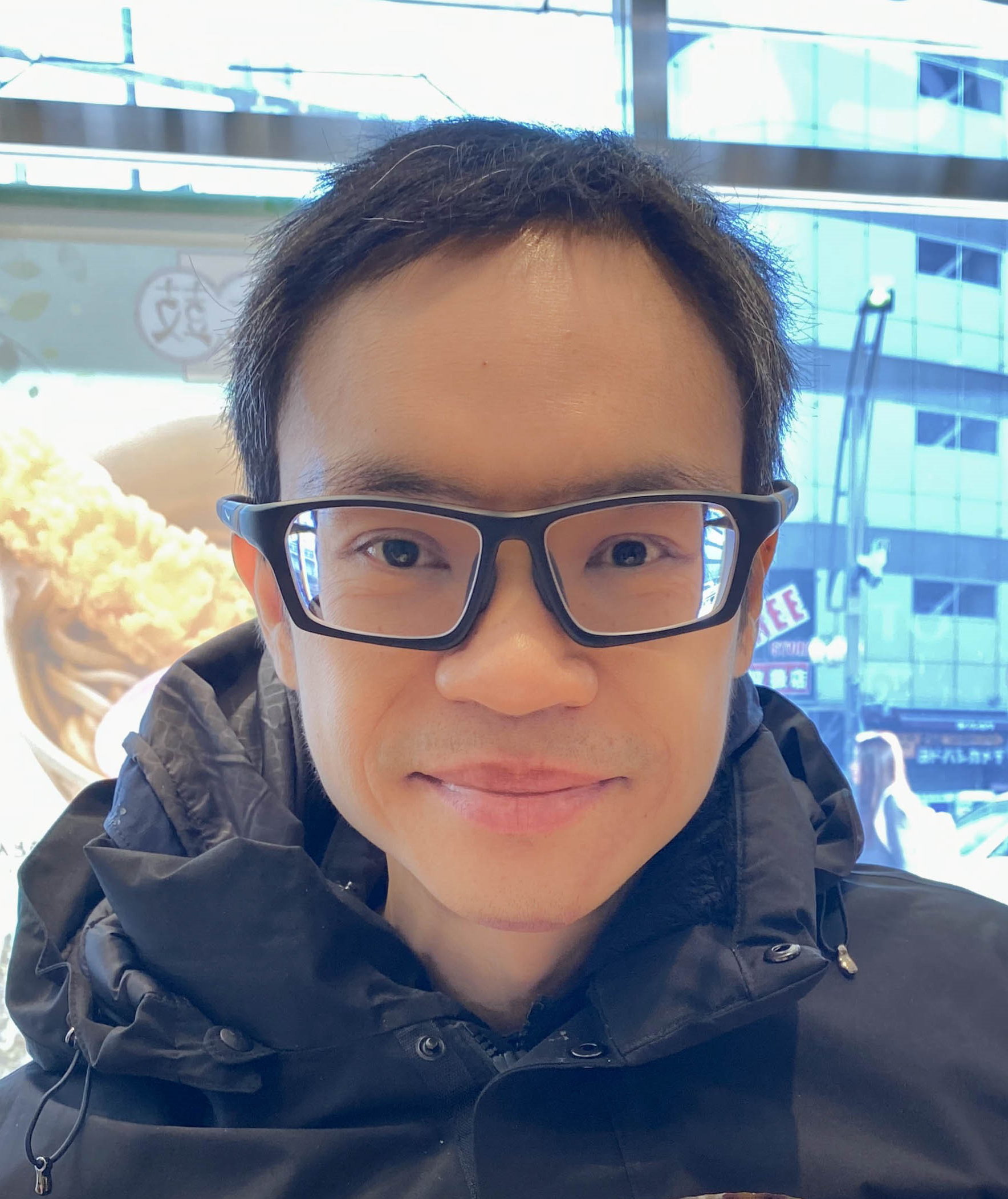}}]{Dusit Niyato}
(Fellow, IEEE) is a professor in the College of Computing and Data Science, at Nanyang Technological University, Singapore. He received B.Eng. from King Mongkuts Institute of Technology Ladkrabang (KMITL), Thailand and Ph.D. in Electrical and Computer Engineering from the University of Manitoba, Canada. His research interests are in the areas of sustainability, edge intelligence, decentralized machine learning, and incentive mechanism design.
\end{IEEEbiography}

\end{document}